\documentclass[
 aip,
 sd,%
 amsmath,amssymb,
preprint,%
nofootinbib]
{revtex4-1}
\usepackage{graphicx}
\usepackage{dcolumn}
\usepackage{bm}
\usepackage{amsthm}
\newtheorem{theorem}{Theorem}
\newtheorem{corollary}{Corollary}

\begin{document}

\title[Lattice Sum Zeros:2]{Zeros of Lattice Sums: 2. A Geometry for the Generalised Riemann Hypothesis}

\author{R.C. McPhedran,\\
School of Physics, University of Sydney,\\
Sydney, NSW Australia 2006.}
  
\begin{abstract} 
 The location of zeros of the basic double sum over the square lattice is studied. This sum can be represented in terms of the product of the Riemann zeta function and the Dirichlet beta function, so that the assertion that all its non-trivial zeros lie on the critical line is a particular case of the Generalised Riemann Hypothesis (GRH). It is shown that a new necessary and sufficient condition for this special case of the GRH to hold is that a particular set of equimodular and equiargument contours of a ratio of MacDonald function double sums intersect only on the critical line. It is further shown that these contours could only intersect off the critical line on the boundary of  discrete regions of the complex plane called "inner islands". Numerical investigations are described related to this geometrical condition, and it is shown that for the first ten thousand zeros of both the zeta function and the beta function over 70\% of zeros lie outside the inner islands, and thus would be guaranteed to lie on the critical line by the arguments presented here. A new sufficient condition for the Riemann Hypothesis to hold is also presented.
\end{abstract}
\maketitle





%


\section{Introduction}
The Riemann Hypothesis (RH) that all non-trivial zeros of the function $\zeta(s)$ lie on the critical line $\Re(s)=\Re(\sigma+i t)=1/2$ is widely regarded as one of the most important and difficult unsolved problems in mathematics\cite{tandhb}. The Generalised Riemann Hypothesis (GRH) that non-trivial zeros of Dirichlet $L$ functions
with integer characters also lie on the critical line has also been widely investigated. The results we present below consist of a number of numerical and analytic investigations of a particular case of the GRH, pertaining to the most important double sum of the Epstein zeta type:
\begin{equation}
S_0(s;\lambda)=\sum_{p_1,p_2}' \frac{1}{(p_1^2+p_2^2 \lambda^2)^s},
\label{et1}
\end{equation}
where the sum over the integers $p_1$ and $p_2$ runs over all integer pairs, apart from $(0,0)$, as indicated by the superscript prime. The quantity $\lambda$ corresponds to the period ratio of the
rectangular lattice, and $s$ is an arbitrary complex number. For $\lambda^2$ an integer, this is an Epstein zeta function, but for $\lambda^2$ non-integer we will refer to it as a lattice sum over the rectangular lattice. Many results connected with lattice  sums of this and more general forms have been collected in the recent book {\em Lattice Sums Then and Now}\cite{lsb}, hereafter denoted {\em LSTN}. For $\lambda=1$, the sum (\ref{et1}) takes a simple form for which the GRH is applicable:
\begin{equation}
S_0(s;1)=4 \zeta(s) L_{-4}(s),
\label{et1a}
\end{equation}
using the notation of Zucker and Robertson \cite{zandr} for Dirichlet $L$ functions.
For $\lambda\neq 1$, in general $S_0(\lambda, s)$ will have non-trivial zeros off the critical line, as was discussed in a previous article\cite{part1} (hereafter referred to as I).

Bogomolny and  Leboeuf\cite{bandl} have  discussed the distribution and  separation of zeros for $S_0( s;1)$, finding that the product form (\ref{et1a}) of this basic sum resulted in
a distribution of zeros with higher probability of smaller gaps than for individual Dirichlet $L$ functions.
Numerical investigations of the distribution and separation of zeros of  more general Epstein zeta functions have been discussed by Hejhal\cite{hejhal1}, and by Bombieri and Hejhal \cite{hejhal2}. Such investigations are difficult for large $t$ even on the most powerful available computers, due to the number of terms required in the most convenient general expansion for the functions (see Section 2) and the degree of cancellation between terms.

In this paper, we will concentrate on the case of the square lattice ($\lambda=1$), but will also use results from I\cite{part1} in the limit as $\lambda\rightarrow 1$. We hope to demonstrate that the context of double sums and rectangular lattices is richer than that of single sums like that in the Riemann zeta function, and that this greater richness offers extra opportunities for the development of analytic arguments relating to the RH and the GRH. The results will be accompanied by proofs which may not attain the fullest contemporary degree of rigour, but which may hopefully inspire other workers to remedy this defect. The results have been obtained on the basis of extensive numerical investigations, and some typical graphical examples will be presented. It should be stressed that it is not overly difficult for the expressions  presented below to be employed in appropriate symbolic software by those interested in their own explorations of the geometric context we describe.

Section 2 contains essential results from I, both in their form for general $\lambda$ and for $\lambda=1$. These are used in Sections 3 and 4 to prove significant results for double sums, including the division of the complex $s$ plane into extended regions (running from $\sigma=-\infty$ to $\sigma=\infty$), discrete island regions (with bounded variation in $\sigma$ and $t$) and inner island regions within the latter. The most important result is that all zeros of $S_0(s;1)$ not lying on the boundaries between island and inner island regions
must lie on the critical line. Graphs are given in Section 4 showing typical configurations of the three regions, and tabular data is given in Section 5 on the distribution various types of zeros among the three regions.

\section{Rectangular and Square Lattice Sums}
The double sums we consider are, for the rectangular lattice, analytic in the complex variable $s$, and depend on the real parameter $\lambda$. They reduce to sums for the square lattice when $\lambda$ tends to unity. For brevity of notation, we will sometimes omit the parameter $\lambda$ when it takes the value unity. We will also indicate the partial derivative with respect to $\lambda$ by attaching this symbol as a subscript to the function name.

Connected to the double sum  (\ref{et1}) is a  general class of MacDonald function double sums for rectangular lattices:
\begin{equation}
{\cal K}(n,m;s;\lambda)=\pi^n\sum_{p_1,p_2=1}^\infty  \left(\frac{p_2^{s-1/2+n}}{p_1^{s-1/2-n}}\right) K_{s-1/2+m}(2\pi p_1 p_2\lambda).
\label{mac1}
\end{equation}
For $\lambda\geq 1$ and the (possibly complex) number $s$ small in magnitude, such sums converge rapidly, facilitating numerical evaluations. (The sum gives accurate answers
as soon as the argument of the MacDonald function exceeds the modulus of its order by  a factor of 1.3 or so.) The double sums satisfy the following symmetry relation, obtained by interchanging $p_1$ and $p_2$ in the definition (\ref{mac1}):
\begin{equation}
{\cal K}(n,-m;s;\lambda)={\cal K}(n,m;1-s;\lambda).
\label{mac1a}
\end{equation}

The lowest order sum ${\cal K}(0,0;s;\lambda)$ occurs in the representation of $S_0(s;\lambda)$ due to Kober\cite{kober}:
\begin{equation}
\lambda^{s+1/2} \frac{\Gamma(s)}{8\pi^s} S_0(\lambda;s)=\frac{1}{4} \frac{\xi_1(2 s)}{\lambda^{s-1/2}}+\frac{1}{4} \lambda^{s-1/2} \xi_1(2 s-1)+{\cal K}(0,0;s;\frac{1}{\lambda}).
\label{mac2}
\end{equation}
Here $\xi_1(s)$ is the symmetrised zeta function. 
In terms of the Riemann zeta function, (\ref{mac2}) is
\begin{equation}
S_0(s;\lambda)=\frac{2 \zeta (2 s)}{\lambda^{2 s}}+2\sqrt{\pi}\frac{\Gamma(s-1/2) \zeta(2 s-1)}{\Gamma(s)\lambda}+
\frac{8\pi^s}{\Gamma(s) \lambda^{s+1/2}}{\cal K}(0,0;s;\frac{1}{\lambda}).
\label{mac2a}
\end{equation}

A fully symmetrised form of (\ref{mac2}) (symmetric under both $s\rightarrow 1-s$ and $\lambda\rightarrow1/\lambda$)  is:
\begin{equation}
\lambda^{s} \frac{\Gamma(s)}{8\pi^s} S_0(s;\lambda)={\cal T}_+(s;\lambda)+\frac{1}{\sqrt{\lambda}}  {\cal K}(0,0;s;\frac{1}{\lambda}),
\label{mac2s}
\end{equation}
where
\begin{equation}
{\cal T}_+(s;\lambda)=\frac{1}{4}\left[\frac{\xi_1(2 s)}{\lambda^s}+\frac{\xi_1(2 s-1)}{\lambda^{1-s}}\right].
\label{mac2s1}
\end{equation}
Note that ${\cal T}_+( 1-s;\lambda)={\cal T}_+(s;\lambda)$ and $ {\cal K}(0,0;1-s;\lambda)={\cal K}(0,0;s;\lambda)$, so that the left-hand side of equation (\ref{mac2s}) must then be unchanged under replacement of $s$ by $1-s$. The left-hand side is also unchanged under replacement of $\lambda$ by $1/\lambda$, so the same is true for the
sum of the two terms on the right-hand side, although in general it will not be true for them individually. The symmetry relations for $S_0(s;\lambda)$ then are
\begin{equation}
\lambda^{s} \frac{\Gamma(s)}{8\pi^s} S_0(s;\lambda)=\frac{1}{\lambda^{s}} \frac{\Gamma(s)}{8\pi^s} S_0\left(s;\frac{1}{\lambda}\right)=
\lambda^{1-s} \frac{\Gamma(1-s)}{8\pi^{(1-s)}} S_0(1-s;\lambda)=\frac{1}{\lambda^{1-s}} \frac{\Gamma(1-s)}{8\pi^{(1-s)}} S_0\left(1-s;\frac{1}{\lambda}\right).
\label{mac2s2}
\end{equation}
From the equations (\ref{mac2s2}), if $s_0$ is a zero of $S_0(s;\lambda)$ then
\begin{equation}
S_0(s_0;\lambda)=0~ \implies~S_0(s_0;1/\lambda)=0=S_0(1-s_0;1/\lambda)=S_0(1-s_0;\lambda).
\label{mac2s2a}
\end{equation}
Another interesting deduction from (\ref{mac2s}) relates to the derivative of $S_0(\lambda, s_0)$ with respect to $\lambda$:
\begin{eqnarray}
& &\lambda^s S_0(s;\lambda) =\frac{1}{\lambda^s} S_0\left(s;\frac{1}{\lambda}\right)~\implies~\nonumber \\
&& s \lambda^{s-1} S_0(s;\lambda) +\lambda^s \frac{\partial}{\partial \lambda} S_0(s;\lambda) =\frac{-s}{\lambda^{s+1}} S_0\left(s;\frac{1}{\lambda}\right)-
\frac{1}{\lambda^{s+2}}\frac{\partial}{\partial \lambda}  S_0\left(s;\frac{1}{\lambda}\right),
\label{mac2s2b}
\end{eqnarray}
so that
\begin{equation}
\left. \frac{\partial}{\partial \lambda} S_0(s;\lambda)\right|_{\lambda=1}=S_{0,\lambda}(s;1) =-s S_0(s;1).
\label{mac2s2c}
\end{equation}

Combining (\ref{mac2s}) and (\ref{mac2s2}), we arrive at a general symmetry relationship  for ${\cal K}(0,0;s;\lambda)$:
\begin{equation}
{\cal T}_+(s;\lambda)-{\cal T}_+\left(s; \frac{1}{\lambda} \right)=\sqrt{\lambda}{\cal K}(0,0;s;\lambda)-\frac{1}{\sqrt{\lambda}}  {\cal K}\left(0,0;s;\frac{1}{\lambda}\right),
\label{mac2s3}
\end{equation}
or
\begin{eqnarray}
&&\frac{1}{4}\left[\xi_1(2 s)\left(\frac{1}{\lambda^{s}}-\lambda^s\right) +\xi_1(2 s-1)\left(\frac{1}{\lambda^{1-s}}-\lambda^{1-s}\right) \right]=\nonumber\\
&&\sqrt{\lambda}{\cal K}(0,0;s;\lambda)-\frac{1}{\sqrt{\lambda}}  {\cal K}\left(0,0;s;\frac{1}{\lambda}\right).
\label{mac2s4}
\end{eqnarray}
This identity holds for all values of $s$ and $\lambda$. One use of it is to expand about $\lambda=1$, which gives identities for the partial derivatives of 
${\cal K}(0,0;s;\lambda)$ with respect to $\lambda$, evaluated  at $\lambda=1$. The first of these is
\begin{equation}
{\cal L}(s)=s\xi_1(2 s)+(1-s) \xi_1(2 s-1)=-2 {\cal K}(0,0;s;1)-4 {\cal K}_\lambda(0,0;s;1).
\label{mac2s5}
\end{equation}
All three functions occurring in (\ref{mac2s5}) are even under $s\rightarrow 1-s$.

By analogy to  the equation (\ref{mac2s1}) we define:
\begin{equation}
{\cal T}_-(s;\lambda)=\frac{1}{4}\left[\frac{\xi_1(2 s)}{\lambda^s}-\frac{\xi_1(2 s-1)}{\lambda^{1-s}}\right].
\label{mac2s1a}
\end{equation}
This function is odd under $s\rightarrow 1-s$.

It is known\cite{prt,hejhal3,ki,lagandsuz,mcp13} that ${\cal T}_+(s;\lambda)$ and ${\cal T}_-(s;\lambda)$ have all their zeros on the critical line if $\lambda\leq 1$ and $t> 3.9125$. This can be easily seen from the properties of the function
\begin{equation}
{\cal U}(s)=\frac{\xi_1(2 s-1)}{\xi_1(2 s)},
\label{Udef}
\end{equation}
which has modulus smaller than unity to the right of the critical line (where its numerator has its zeros)  and less than unity to its left (where its denominator has its zeros) for  $t> 3.9125$. By contrast, ${\cal K}(0,0;s)$ has zeros both on the critical line and off it\cite{mcp04}. Also, Lagarias and Suzuki\cite{lagandsuz} have proved that the function we denote by ${\cal L}(s)$ has all its zeros on the critical line. ( We may understand this result  since ${\cal L}(s)=0$ if and only if $-1={\cal U}(s) (1-s)/s$. The modulus of the right-hand side is smaller than unity for $\sigma>1/2$, and larger than unity for $\sigma<1/2$.)

In addition to the function ${\cal U}(s)$, we will employ a closely associated function
\begin{equation}
{\cal V}(s)=\frac{{\cal T}_+(s)}{{\cal T}_-(s)}=\frac{1+{\cal U}(s)}{1-{\cal U}(s)}.
\label{Vdef}
\end{equation}
${\cal V}(s)$ is purely imaginary on the critical line and ${\cal U}(s)$ has modulus unity there. The fixed points of the transformation (\ref{Vdef}) are ${\cal V}(s)={\cal U}(s)=\pm i$,
and its normal form is
\begin{equation}
\frac{{\cal V}(s)-i}{{\cal V}(s)+i}=i\left(\frac{{\cal U}(s)-i}{{\cal U}(s)+i}\right).
\label{VUform}
\end{equation}
\section{Properties Related to ${\cal K}(1,1;s;\lambda)$}
The recurrence relations for MacDonald functions give rise to those for the double sums:
\begin{equation}
\frac{\partial}{\partial \lambda} {\cal K}(n,m;s;\lambda)=-[{\cal K}(n+1,m+1;s;\lambda)+{\cal K}(n+1,m-1;s;\lambda)],
\label{mac3}
\end{equation}
and
\begin{equation}
\frac{(m+s-1/2)}{\lambda} {\cal K}(n,m;s;\lambda)=[{\cal K}(n+1,m+1;s;\lambda)-{\cal K}(n+1,m-1;s;\lambda)].
\label{mac4}
\end{equation}
These may be used to construct operators which raise $n$ and lower $m$, or raise $n$ and raise $m$, respectively:
\begin{equation}
-\frac{1}{2}\left[ \frac{\partial}{\partial \lambda}+\frac{(m+s-1/2)}{\lambda}  \right] {\cal K}(n,m;s;\lambda)= {\cal K}(n+1,m-1;s;\lambda),
\label{mac5}
\end{equation}
and
\begin{equation}
-\frac{1}{2}\left[ \frac{\partial}{\partial \lambda}-\frac{(m+s-1/2)}{\lambda}  \right] {\cal K}(n,m;s;\lambda)= {\cal K}(n+1,m+1;s;\lambda).
\label{mac5a}
\end{equation}

From (\ref{mac5a}) we have
\begin{equation}
{\cal K}(1,1;s;\lambda)=-\frac{1}{2}{\cal K}_\lambda (0,0;s;\lambda)+\frac{(s-1/2)}{2 \lambda}{\cal K}(0,0;s;\lambda)
\label{K11def}
\end{equation}
The symmetric and antisymmetric parts of ${\cal K}(1,1;s;\lambda)$ are
\begin{equation}
{\cal K}(1,1;s;\lambda)+{\cal K}(1,1;1-s;\lambda)=-{\cal K}_\lambda (0,0;s;\lambda),
\label{K11sym}
\end{equation}
and
\begin{equation}
{\cal K}(1,1;s;\lambda)-{\cal K}(1,1;1-s;\lambda)=\frac{(s-1/2)}{ \lambda}{\cal K}(0,0;s;\lambda).
\label{K11asym}
\end{equation}

It is  useful to define
\begin{equation}
{\cal V}_{\cal K} (1,1;s;\lambda)=\frac{{\cal K}(1,1;s;\lambda)-{\cal K}(1,1;1-s;\lambda)}{{\cal K}(1,1;s;\lambda)+{\cal K}(1,1;1-s;\lambda)},
\label{etf24}
\end{equation} 
and
\begin{equation}
{\cal U}_{\cal K} (1,1;s;\lambda)=\frac{{\cal K}(1,1;s;\lambda)}{{\cal K}(1,1;1-s;\lambda)}=\left(\frac{1+{\cal V}_{\cal K} (1,1;s;\lambda)}{1-{\cal V}_{\cal K} (1,1;s;\lambda)}\right).
\label{etf24a}
\end{equation}
From equations (\ref{K11sym},\ref{K11asym}),
\begin{equation}
{\cal V}_{\cal K} (1,1;s;\lambda)=\frac{-(s-1/2){\cal K}(0,0;s;\lambda)}{\lambda {\cal K}_\lambda (0,0;s;\lambda)}=\frac{-(s-1/2)}{\lambda \partial {\log\cal K}(0,0;s;\lambda)/\partial \lambda}.
\label{etf25}
\end{equation}
From equations (\ref{etf24a}) and (\ref{etf25}), we have the symmetry relations
\begin{equation}
{\cal U}_{\cal K} (1,1;1-s;\lambda)=\frac{1}{{\cal U}_{\cal K} (1,1;s;\lambda)}, ~{\cal V}_{\cal K} (1,1;1-s;\lambda)=-{\cal V}_{\cal K} (1,1;s;\lambda).
\label{symuv}
\end{equation}

In what follows, we will abbreviate the notation for the sums ${\cal K}$ and their $\lambda$ derivatives by suppressing the entry for the geometric parameter $\lambda$ when it takes the value unity. We will do the same for ${\cal S}_0(s;\lambda)$.

\section{ ${\cal K}(1,1;s)$ and the Zeros of $S_0(s)$}
From equations (\ref{K11sym}, \ref{K11asym}, \ref{mac2s}) we find:
\begin{equation}
{\cal K}(1,1;s)= s\left[  \frac{\Gamma(s)}{16\pi^s} S_0(s)\right]-(s-1/2) \left[\frac{\xi_1(2s-1)}{4} \right] .
\label{n26}
\end{equation}
The functions ${\cal U}_{\cal K} (1,1;s;\lambda)$, ${\cal V}_{\cal K} (1,1;s;\lambda)$ become in the special case $\lambda=1$
\begin{equation}
{\cal U}_{\cal K} (1,1;s)=\left[\frac{ \frac{1}{( s-1/2)}  \log {\cal K}_\lambda(0,0;s)-1}
{ \frac{1}{( s-1/2)}  \log {\cal K}_\lambda(0,0;s)+1}\right],
\label{srep4}
\end{equation}
or equivalently
\begin{equation}
{\cal U}_{\cal K} (1,1;s)=\left(\frac{s}{1-s}\right)\frac{\left[\frac{\Gamma(s) C(0,1;s)}{\pi^s}-4\left(1-\frac{1}{2 s}\right)\xi_1(2 s-1)\right]}
{\left[\frac{\Gamma(s) C(0,1;s)}{\pi^s}-4\left(1-\frac{1}{2-2  s}\right)\xi_1(2 s)\right]},
\label{srep4a}
\end{equation}
and
\begin{equation}
{\cal V}_{\cal K} (1,1;s;1)=\frac{-(s-1/2){\cal K}(0,0;s)}{ {\cal K}_\lambda(0,0;s)}
=\frac{-(s-1/2)}{ \log{\cal K}_\lambda(0,0;s)}.
\label{srep5}
\end{equation}
or
\begin{equation}
{\cal V}_{\cal K} (1,1;s;1)=\frac{\Gamma (s) S_0(1,s)/(2 \pi^s)-[\xi_1(2 s)+\xi_1(2 s-1)]}{\Gamma (s) S_0(1,s)/(4 \pi^s (s-1/2))+[\xi_1(2 s)-\xi_1(2 s-1)]}.
\label{srep6}
\end{equation}
Returning to the equation (\ref{mac2s5}), we know\cite{lagandsuz} that its left-hand side has no zeros off the critical line. From the right-hand side, the equivalent statement is
that
\begin{equation}
\frac{ {\cal K}_\lambda(0,0;s)}{{\cal K}(0,0;s)}=-\frac{1}{2}\implies \sigma=\frac{1}{2}.
\label{srep9}
\end{equation}
When equation (\ref{srep9}) is satisfied, we have from (\ref{srep4}) that
\begin{equation}
{\cal U}_{\cal K} (1,1;s)=\frac{s}{1-s}=-{\cal U}(s),
\label{srep10}
\end{equation}
so that zeros of the function ${\cal L}(s)$ giving the left-hand side of equation (\ref{mac2s5}) satisfy the same equation (${\cal U}_{\cal K} (1,1;s;1)=-{\cal U}(1,s)$) as do solutions of $S_0(1,s)=0$.

Let
\begin{equation}
{\cal F}(s)=\frac{{\cal U}_{\cal K} (1,1;s)}{{\cal U}(s)}.
\label{mthm1}
\end{equation}
Then the relationship between  ${\cal K}(1,1;s)$ and the zeros of $S_0(s)$ is established in the following result.
\begin{theorem}
If $S_0(s_0)=0$ then ${\cal F}(s_0)=-1$. If ${\cal F}(s_0)=-1$ then either ${\cal L}(s_0)=0$, in which case $s_0$ must lie on the critical line, or $S_0(s_0)=0$.
\label{mthm}
\end{theorem}
\begin{proof}
At a zero $s_0$ of $S_0(s)$, we have from (\ref{mac2s}) that ${\cal T}_+(s_0)=-{\cal K}(0,0;s_0)$. For general $s$, from (\ref{mac2s5}),
\begin{equation}
{\cal L}(s)=2{\cal T}_+(s)+2(2 s-1) {\cal T}_-(s)=-2 {\cal K}(0,0;s)-4 {\cal K}_\lambda(0,0;s).
\label{mac2s5a}
\end{equation}
Hence, ${\cal K}_\lambda(0,0;s_0)=-(s_0-1/2){\cal T}_-(s_0)$. From (\ref{etf25}), 
\begin{equation}
{\cal V}_{\cal K} (1,1;s_0)=-\frac{{\cal T}_+(s_0)}{{\cal T}_-(s_0)}=-{\cal V}(s_0), 
\label{mthm1a}
\end{equation}
so that from (\ref{etf24a}) and (\ref{Vdef}),
\begin{equation}
{\cal U}_{\cal K} (1,1;s_0)=-{\cal U}(s_0). 
\label{mthm2}
\end{equation}

Consider now the case  ${\cal F}(s_0)=-1$, for which ${\cal V}_{\cal K} (1,1;s_0)=-{\cal V}(s_0)$ . Define $\tilde{S}_0(s)=\Gamma (s) S_0(s)/(8 \pi^s)$. Then from (\ref{etf25}) and (\ref{mac2s}),
\begin{equation}
\tilde{S}_0(s)=\frac{{\cal V}_{\cal K} (1,1;s){\cal T}_-(s)+{\cal T}_+(s)}{1-{\cal V}_{\cal K} (1,1;s)/(2 s-1)}=\frac{{\cal T}_-(s)({\cal V}_{\cal K} (1,1;s)+{\cal V}(s))}{1-{\cal V}_{\cal K} (1,1;s)/(2 s-1)}.
\label{mthm3}
\end{equation}
This can also be written as
\begin{equation}
\tilde{S}_0(s)=\frac{2 (2 s-1) {\cal T}_-(s)^2 [{\cal V}_{\cal K} (1,1;s)+{\cal V}(s)]}{{\cal L}(s)-2 {\cal T}_-(s)[{\cal V}_{\cal K} (1,1;s)+{\cal V}(s)]}.
\label{mthm4}
\end{equation}
Hence ${\cal V}_{\cal K} (1,1;s_0)=-{\cal V}(s_0)$ guarantees  $\tilde{S}_0(s)=0$, unless ${\cal L}(s_0)=0$, in which case $\tilde{S}_0(s_0)=- (2 s_0-1) {\cal T}_-(s_0)$.
Note that if ${\cal L}(s_0)=0$, from (\ref {mac2s5a}) we know that $ {\cal T}_+(s_0)\neq 0$ and  $ {\cal T}_-(s_0)\neq 0$, since these two functions have no zeros in common.
Hence, if  ${\cal L}(s_0)=0$ then $\tilde{S}_0(s_0)\neq 0$.
\end{proof}

Another easily proved result is the following:
\begin{theorem} ${\cal K}(1,1;s)$ and  ${\cal K}(1,1;1-s)$ are not simultaneously zero for $s$ off the critical line. ${\cal K}(0,0;s)$ and  ${\cal K}_{\lambda}(0,0;s)$ are not simultaneously zero for $s$ off the critical line.
\end{theorem}
\begin{proof}
From equations (\ref{K11sym}) and (\ref{K11asym}), if ${\cal K}(1,1;s)$ and  ${\cal K}(1,1;1-s)$ are both zero for some $s$, then so are ${\cal K}(0,0;s)$ and  ${\cal K}_{\lambda}(0,0;s)$.
From (\ref{mac2s5a}) we then have ${\cal L}(s)=0$, so that $s$ must lie on the critical line.
\end{proof}
Note that if ${\cal K}(1,1;s)$ and  ${\cal K}(1,1;1-s)$ are equal, then ${\cal K}(0,0;s)$ has to be zero, which can occur for $s$ either on or off the critical line. If ${\cal K}(1,1;s)$ and  $-{\cal K}(1,1;1-s)$ are equal, then ${\cal K}_\lambda (0,0;s)$ has to be zero. It is suspected that this can only occur for $s$ on the critical line, but a proof of this would be valuable.

We next consider that behaviour of  ${\cal U}_{\cal K}(1,1;s)$ in the complex plane for $t$ not small. The theorem which follows shows that this function well away from the critical line
has the opposite behaviour to ${\cal U}(s)$. The latter is smaller than unity in magnitude to the right of the critical line, and larger in magnitude than unity to the left of it. The former has magnitude which increases without bound for $\sigma$ moving well to the right of the critical line, and tends towards zero as $\sigma$ moves well to the left of the critical line.

 \begin{theorem} The function ${\cal U}_{\cal K}(1,1;s)$ has "island" regions defined by boundaries inside which its modulus is less than unity, and outside which it is greater than unity.
  It has  monotonic argument variation around each side of island regions surrounding intervals of the critical line. 
 \label{thmisland}
 \end{theorem}
 \begin{proof}
 We have  from equation (\ref{n26}) the expansion 
 \begin{equation}
 {\cal U}_{\cal K}(1,1;s)=\frac{s {\cal C}(0,1;s)-4\sqrt{\pi}\Gamma (s-1/2) \zeta (2 s-1)/\Gamma (s)}{(1-s) {\cal C}(0,1;s)+4 (s-1/2) \zeta (2 s)}.
 \label{isleq-1}
 \end{equation}
 Using the asymptotic expansion for $\Gamma (s+1/2)/\Gamma (s)$ when $|t|>>1$,
  \begin{equation}
 {\cal U}_{\cal K}(1,1;s)\simeq  \frac{s \zeta (s) L_{-4}(s)-\sqrt{\pi s}(1-1/(8 s)+\ldots)  \zeta (2 s-1)}{(1-s) \zeta(s) L_{-4}(s) + (s-1/2) \zeta (2 s)}.
 \label{isleq-2}
 \end{equation}
 We assume $t$ and $\sigma -1/2$ are sufficiently large so the second term in the numerator is negligible compared with the first, and that in the denominator the series for the product $\zeta (s) L_{-4} (s)$
 and for $\zeta (2 s)$ may be used. We then obtain the following approximation from (\ref{isleq-2}):
 \begin{equation}
 {\cal U}_{\cal K}(1,1;s)\simeq  \frac{2 s (1+1/2^s-\sqrt{\pi/s})}{1- (s-1) 2^{1-s}}.
 \label{isleq-3}
 \end{equation}
 This shows that $|{\cal U}_{\cal K}(1,1;s)|\rightarrow \infty$ as $|s|\rightarrow \infty$ in $\sigma >>1/2$. Hence, regions with $|{\cal U}_{\cal K}(1,1;s)|>1$ in $\sigma>1/2$ must be bounded in their $\sigma$
 range. Since $|{\cal U}_{\cal K}(1,1;\overline{1-s})|=|1/{\cal U}_{\cal K}(1,1;s)|$, regions with $|{\cal U}_{\cal K}(1,1;s)|<1$ in $\sigma<1/2$ must also be bounded in their $\sigma$
 range. If they are also limited in their $t$ range, they form islands symmetric about the critical line, with boundaries given by $|{\cal U}_{\cal K}(1,1;s)|=1$. Note that, if $t$ is sufficiently large
 for $O(1/\sqrt{t})$ to be negligible, no zero $s_0$ of $\zeta(2 s-1)$ can lie on a boundary line, since we have  ${\cal U}_{\cal K}(1,1;s_0)=s_0/(1-s_0)$.
 
 On the island boundary in $\sigma>1/2$, $|{\cal U}_{\cal K}(1,1;s)|$ goes from smaller  than unity in the island to larger than unity outside it, and so by the Cauchy-Riemann equations its argument must increase
 around the boundary in the direction of increasing $t$.  Since the argument of this function is even under $s\rightarrow \overline{1-s}$, it must also increase around the left boundary as $t$ increases. On the critical line within the island region, the argument increases as $t$ decreases.
 \end{proof}
 
 A convenient criterion for deciding whether an interval on the critical line is in an extended region or an island region is that
\begin{equation}
\frac{d}{d t} {\cal U}_{\cal K} (1,1;\frac{1}{2}+i t;1)<0 ~{\rm in~ an~ island~ region}
\label{srep15}
\end{equation}
and is positive in an extended region.

\begin{figure}[h]
\includegraphics[width=2.5in]{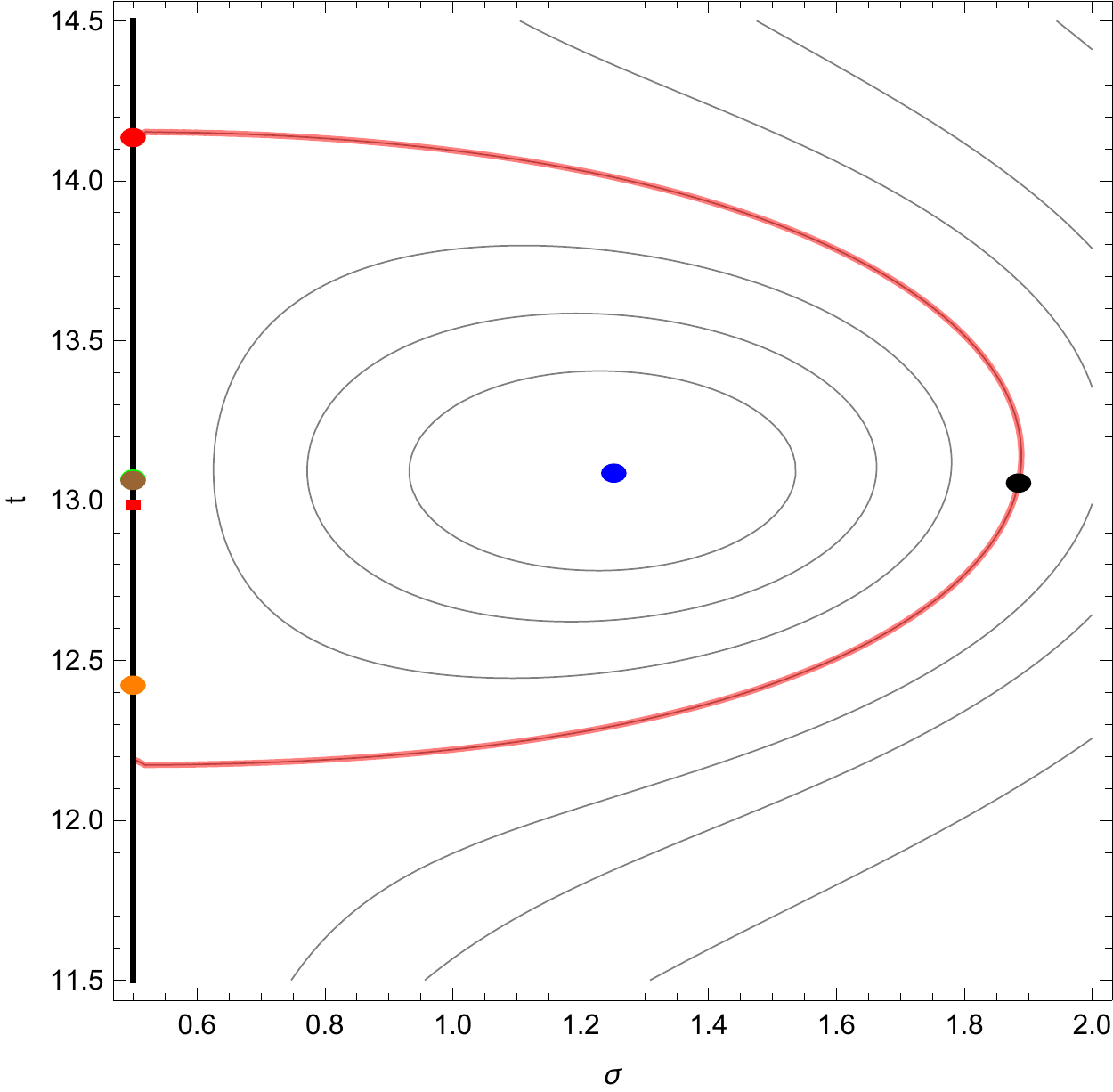}~~\includegraphics[width=2.5in]{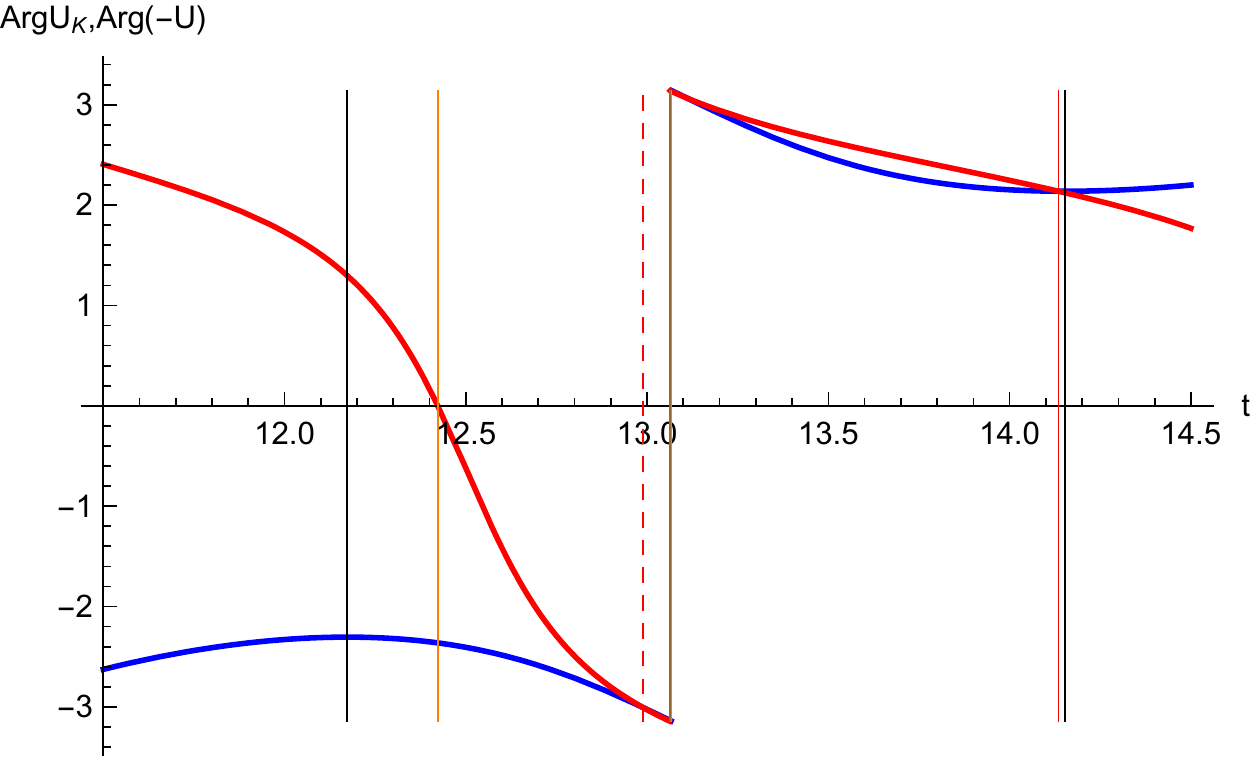}
\caption{(Left) Contours of $|{\cal  U}_{\cal K}(1,1; \sigma+i t)|$ in the plane $(\sigma, t)$, with the red contour corresponding to modulus unity.  The coloured dots correspond to
zeros of $\zeta (s)$ (red), $L_{-4}(s)$ (red rectangle), ${\cal  U}_{\cal K}(0,0; \sigma+i t)$ (black),  ${\cal  U}_{\cal K}(1,1; \sigma+i t)$ (blue), $ {\cal K}_\lambda(0,0; \sigma+i t)$  (green), ${\cal T}_+(1,\sigma+i t)$ (orange) and ${\cal T}_-(1,\sigma+i t)$ (brown). 
(Right)  $\arg{\cal  U}_{\cal K}(1,1; 1/2+i t,1)$ (blue curve) and $\arg (-{\cal U}(1/2+i t))$ (red curve) as a function of $t$, with coloured lines indicating values of $t$ as at left
(the red dashed line corresponding to $L_{-4}$), and the black lines representing the start and end $t$ values of island 1. }
\label{figisland1}
\end{figure}

We give two examples of island regions in Figs. \ref{figisland1}, \ref{figisland2}. Fig. \ref{figisland1} shows the first island region, which extends from $t=12.1731$ to $t=14.1520$.
This region has a simple structure, typical of those observed for higher values of $t$. Values of key points indicated in Fig. \ref{figisland1} are given in Table \ref{tabisle1}.
\begin{table}
\begin{tabular}{|c|c|c|}
\hline
Function for Zero& Zero Value & $\arg({\cal U})$\\ \hline
${\cal K}(0,0)$ & 1.8847+13.0547 i & na \\
${\cal K}(1,1)$ & 1.25182+13.0856 i& na \\ \hline
${\cal T}_+$ & 0.5+12.4226 i & $\pi$\\
${\cal T}_-$ & 0.5+13.0625 i & 0 \\
${\cal K}_\lambda (0,0)$ & 0.5+13.0672 i & -0.0077\\
${\cal L}$ & 0.5+13.1108 i & -0.07624\\ 
$\zeta$ & 0.5+14.1347 i & -1.00321\\ 
$L_{-4}$ & 0.5+13.1108 i & 0.13634\\
\hline
\end{tabular}
\caption{Key points for the first island region.}
\label{tabisle1}
\end{table}

The part of the island region in $\sigma\ge 1/2$ contains a single zero of ${\cal K}(1,1;s)$, and so the argument of ${\cal U}_{\cal K}(1,1;s)$ increases monotonically through a range of $2\pi$ as the boundary $\Gamma_+$ formed by the red contour in Fig. \ref{figisland1} (left) completed by the interval of the critical line is traversed in the anti-clockwise sense. At $t=12.1731$ and  $t=14.1520$ ${\cal U}(s)$ has argument values -1.84514 and  -1.01777 respectively. Each of the four functions ${\cal T}_+$, ${\cal T}_-$,
${\cal L}$ and ${\cal K}_\lambda(0,0;s)$ have one zero in the island region, with the last three lying in close proximity. The single zero of 
$ {\cal K}_\lambda(0,0;s)$ occurs at a $t$ value close to that of the  off-axis zero of  ${\cal K}(0,0;s)$, as might be expected (given that the latter function is even under $s\rightarrow 1-s$). The single zero of ${\cal C}(0,1;s)$ is of $\zeta (s)$, and occurs close to the upper end of the island.

\begin{figure}[h]
\includegraphics[width=3.5in]{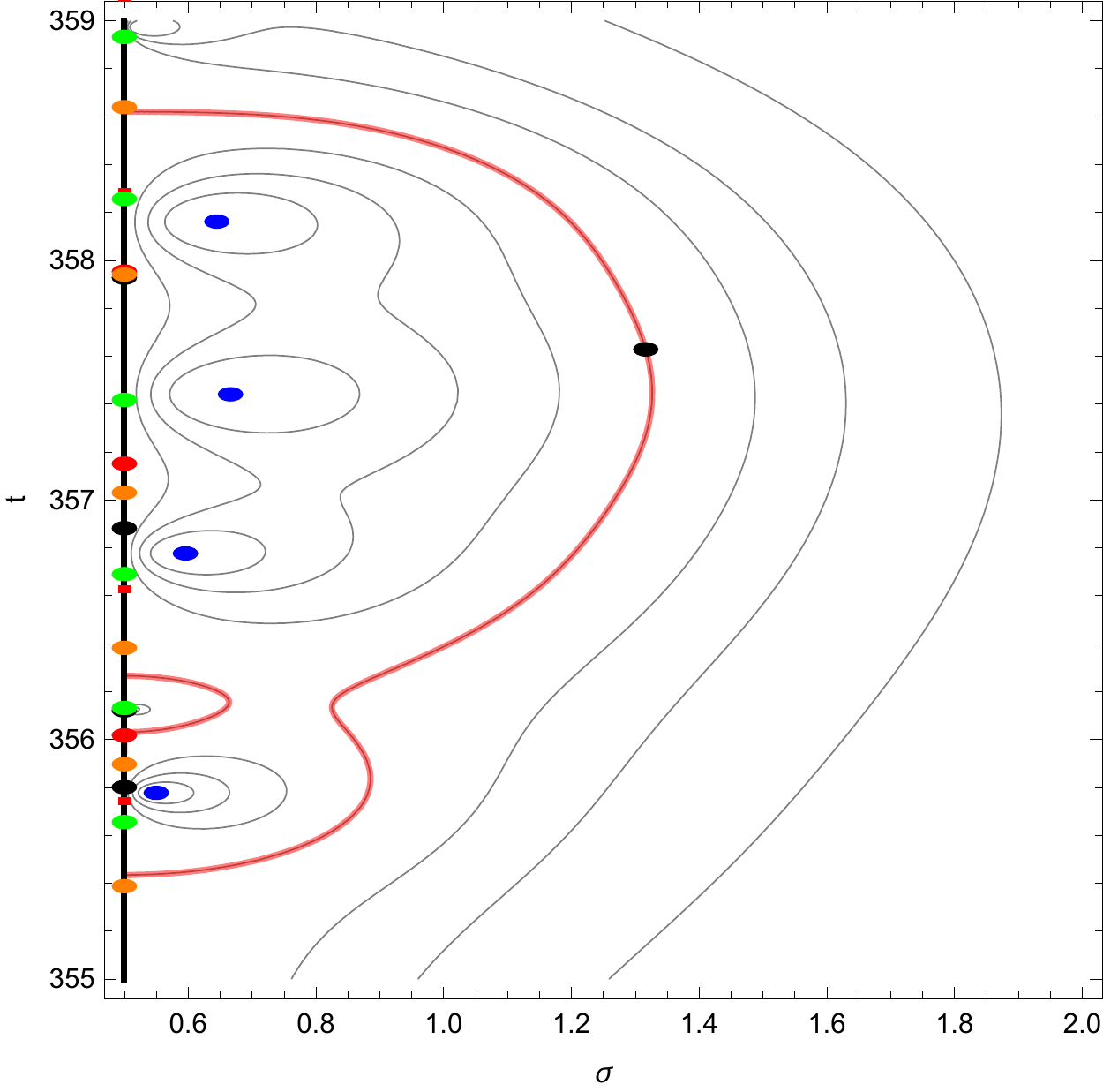}\\
\includegraphics[width=3.5in]{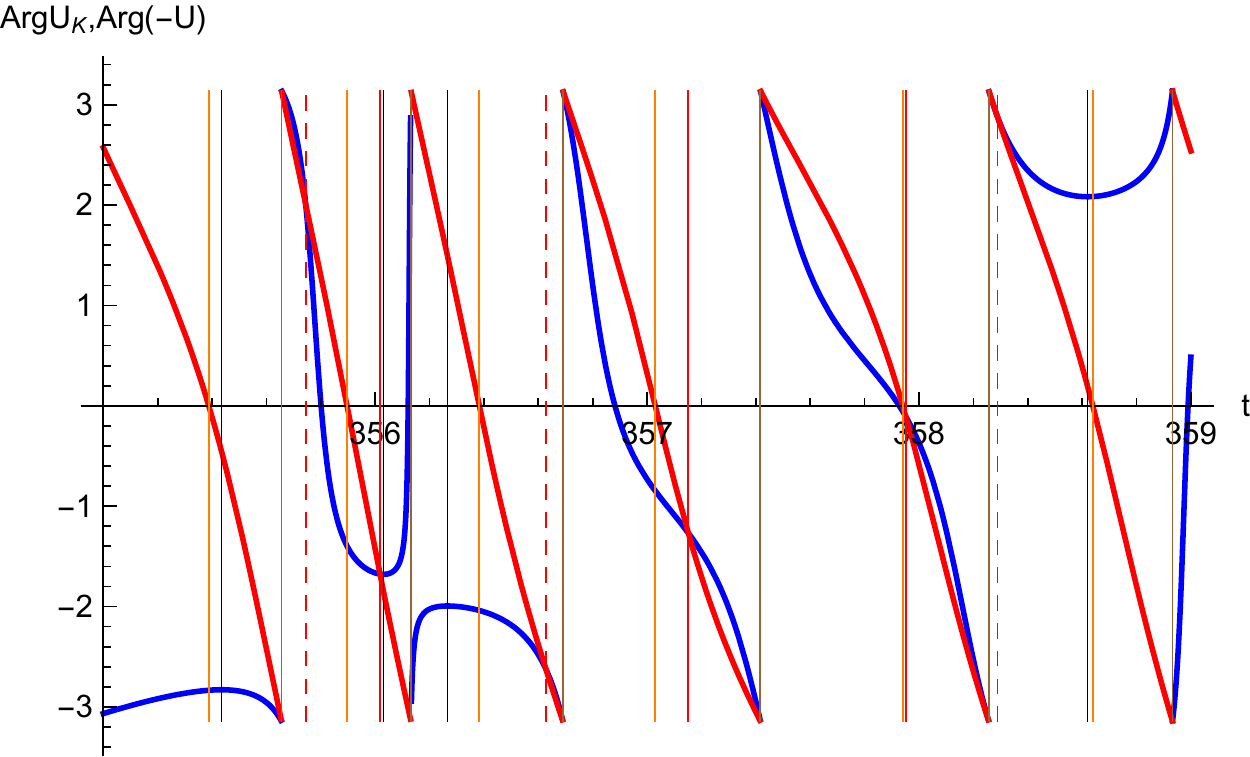}
\caption{(Above) Contours of $|{\cal  U}_{\cal K}(1,1; \sigma+i t)|$ in the plane $(\sigma, t)$, with the red contour corresponding to modulus unity.  The coloured dots correspond to
zeros of $\zeta (s)$ (red), ${\cal  U}_{\cal K}(0,0; \sigma+i t)$ (black),  ${\cal  U}_{\cal K}(1,1; \sigma+i t)$ (blue), $ {\cal K}_\lambda(0,0; \sigma+i t)$  (green), ${\cal T}_+(\sigma+i t)$ (orange) and ${\cal T}_-(\sigma+i t)$ (brown); red rectangles denote zeros of $L_{-4}(s)$. 
(Below)  $\arg{\cal  U}_{\cal K}(1,1; 1/2+i t)$ (blue curve) and $\arg (-{\cal U}(1/2+i t))$ (red curve) as a function of $t$, with coloured lines indicating values of $t$ as above, and the black lines representing the start and end $t$ values of islands 118-119. Red solid lines for chosen values of $t$ indicate zeros of $\zeta(s)$, while red dashed lines indicate zeros of 
$L_{-4}(s)$.}
\label{figisland2}
\end{figure}

For the more complicated case of Fig. \ref{figisland2}, there are four off-axis zeros of ${\cal K}(1,1;s)$ and one pole in the island region for $\sigma>1/2$ (see
Table \ref{tabisle2}).  As a result of the pole, the critical line in the island region is split into three intervals: $t\in (355.4347, 356.0307)$,
$t\in (356.0307,356.2656)$ and $t\in (356.2656, 358.6201)$. In the first and third, $\arg {\cal U}_{\cal K}(1,1;s)$ decreases as $t$ increases, while in the second it increases.
We will call the region including the second interval of the critical line the enclave region; points corresponding to it in Table \ref{tabisle2} are indicated by a superscripted asterisk.
It contains one zero of each of ${\cal K}(0,0;s)$, ${\cal T}_-(s)$, ${\cal K}_\lambda(0,0;s)$, and ${\cal L}(s)$, but none of $S_0(s)$.
The change of argument of $\arg {\cal U}_{\cal K}(1,1;s)$ round the outer boundary of the island region, and along the whole of the critical line from $t=358.6201$ to $t=355.4347$
is then $8\pi-2\pi=6\pi$. There are three zeros for  each of $\zeta (s)$ and $L_{-4}(s)$ lying on the critical line in the whole $t$ interval described. 

In Table \ref{tabisle2}, zeros of ${\cal T}_-(s)$ have been numbered in brackets, and these bracketed numbers have been used to indicate where zeros of ${\cal K}(0,0;s,1)$, ${\cal K}_\lambda(0,0;s)$, $\zeta (s)$ and $L_{-4}(s)$ lie by specifying which is the zero of ${\cal T}_-(s)$ to the left and closest to the zero of the function in question. The tight correlation between zeros of ${\cal T}_-(s)$ and those of ${\cal K}_\lambda(0,0;s)$ will be commented on below. Two zeros of ${\cal K}(0,0;s)$ lie between the second and third zeros of ${\cal T}_-(s)$, adjacent to the enclave region. One zero for each of $\zeta (s)$ and $L_{-4}(s)$ lies to the right of the first zero of ${\cal T}_-(s)$; otherwise, there is only one in subsequent intervals between  zeros of ${\cal T}_-(s)$.
\begin{table}
\begin{tabular}{|c|c|c|c|c|c|}
\hline
Function for Zero& Zero Value & $\arg({\cal U})$ &Function for Zero& Zero Value & $\arg({\cal U})$\\ \hline
${\cal K}(0,0)$ & 1.3164+357.6282 i & na &${\cal K}(1,1)$ & 0.64487+358.1618 i& na \\ 
  & 0.5+355.8009 i(1) &  -1.85894 & &0.6662+ 357.4409 i & \\
  & 0.5+356.1213 i(1)$^*$ &  0.12434 &  &0.5957+ 356.7767i &\\
  & 0.5+356.8817 i (3)& -1.60714 &   &0.5500+ 355.7773 i & \\
  & 0.5+357.9275 i (4) & -3.02227 &    & 0.4951+ 356.1250 i$^*$     & \\
  \hline
${\cal T}_+$ & 0.5+355.8967 i & $\pi$ &${\cal T}_-$ & 0.5+355.6555 i (1)& 0 \\
 & 0.5+356.3824 i &  &  &0.5+ 356.1314 i(2)$^*$ & \\
 & 0.5+357.0303 i &   & &0.5+356.6901 i (3)& \\
 & 0.5+357.9406  i &    & &0.5+357.4165 i(4) &  \\
 & 0.5+358.6393  i &  &  &0.5+ 358.2564 i(5)  & \\
\hline
${\cal K}_\lambda(0,0)$ & 0.5+355.6551 i & 0.0044295 & ${\cal L}$ & 0.5+355.6557 i & -0.00281171\\ 
 &0.5+356.1316 i (2)$^*$ &  -0.00304713 & & 0.5+356.1316 i$^*$ & -0.00280795 \\
 &0.5+ 356.6902  i (3) &-0.000806774 & & 0.5+356.6905 i & -0.00280355\\
 &0.5+357.4167 i (4) & -0.0015263 & & 0.5+357.41698 i & -0.00279785 \\
  &0.5+ 358.2565 i (5) & -0.000235771 & & 0.5+358.25676 i & -0.00279129\\
 \hline
$\zeta$ & 0.5+356.0176 i  (1)& 1.46531  &$L_{-4}$ & 0.5+355.7444 i  (1)&-1.13605  \\
 &0.5+357.1513 i (3) & 1.87538   & & 0.5+356.6277 i  (2) & 0.519007 \\
  & 0.5+357.9527 i (4) & 3.02868  &  & 0.5+358.2883 i (5) & -0.259667 \\
  \hline
\end{tabular}
\caption{Key points for the  island region of Fig. \ref{figisland2}.}
\label{tabisle2}
\end{table}


The most interesting feature evident from Figs. \ref{figisland1} and \ref{figisland2} is the confluence of the argument curves of the functions ${\cal U}_{\cal K}(1,1;s)$ and
$-{\cal U}(s)$ around points where ${\cal T}_-(s)=0$. This confluence can be understood from the discussion around equations (\ref{srep9}) and (\ref{srep10}), given that
when equation (\ref{srep9}) holds,
\begin{equation}
\arg {\cal U}(s)=\arg \left[\frac{1/2+i t}{-1/2+ i t} \right]\sim-\frac{1}{t}~\mod (2 \pi),
\label{srep16}
\end{equation}
while ${\cal T}_-(s)=0$ implies $\arg {\cal U}(s)=0$ modulo $2\pi$.

Table \ref{tabisle3} gives some parameters of the island regions for $t$ ranging up to 500. The number of island regions in each range of 100 in $t$ increases with $t$, and the fraction of the range of $t$ occupied by island regions tends also to increase with $t$, although  the range of $t$ from 200-300 has fewer  islands than one might expect, and the island fraction is higher than for other ranges shown.  The mean length of islands tends to decrease slowly with increasing $t$, but the trend is not marked, particularly bearing in mind the large standard deviations in the distributions of length.
\begin{table}
\begin{tabular}{|c|c|c|c|c|}\hline
Range&Number&Fraction&Mean Length& Standard Deviation Length\\ \hline
0-100& 22& 31\% & 1.42&0.98\\
100-200& 34&38\% &1.19&0.77\\
200-300& 37& 52\% & 1.39&1.90\\
300-400 & 45&41\% &0.90&0.71\\
400-500&49&46\% &0.93&0.91\\ \hline
0-500 & 187 & 41\% & 1.11 & 1.13 \\ \hline 
\end{tabular}
\caption{Statistics of island regions for $t$ ranging up to 500.}
\label{tabisle3}
\end{table}
\begin{figure}[h]
\includegraphics[width=3.5in]{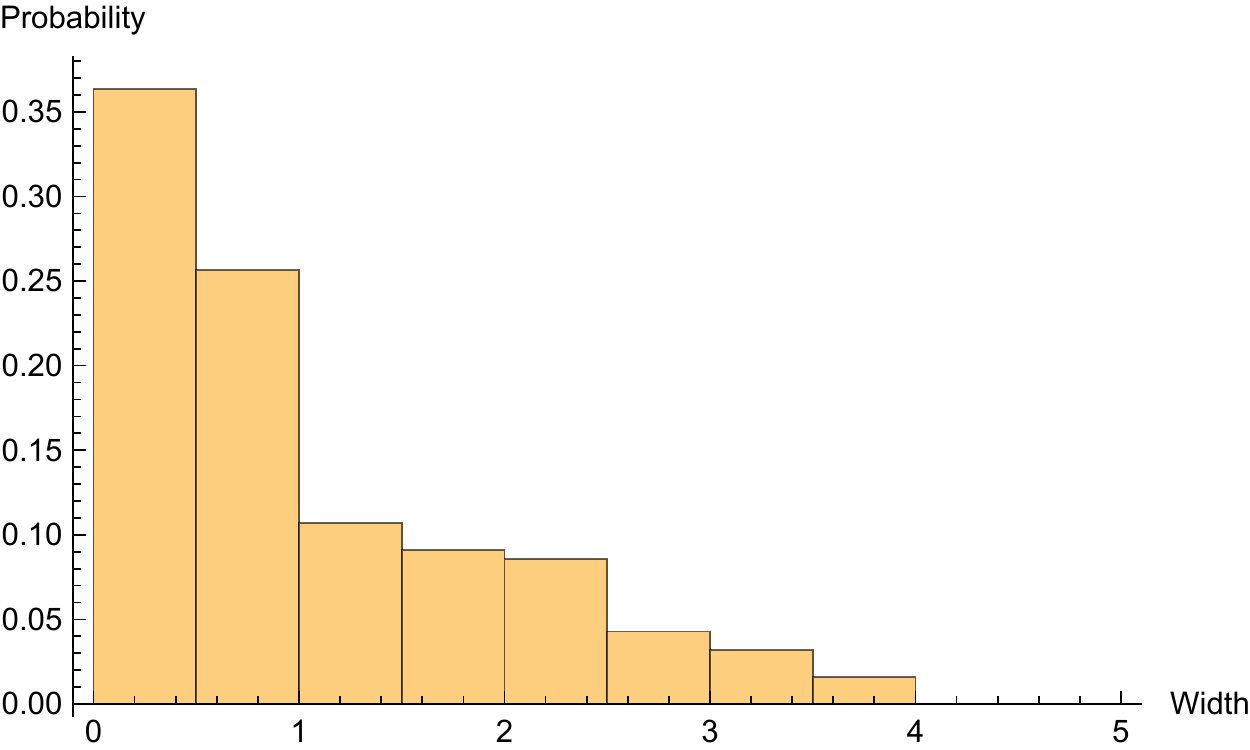}\\
\includegraphics[width=3.5in]{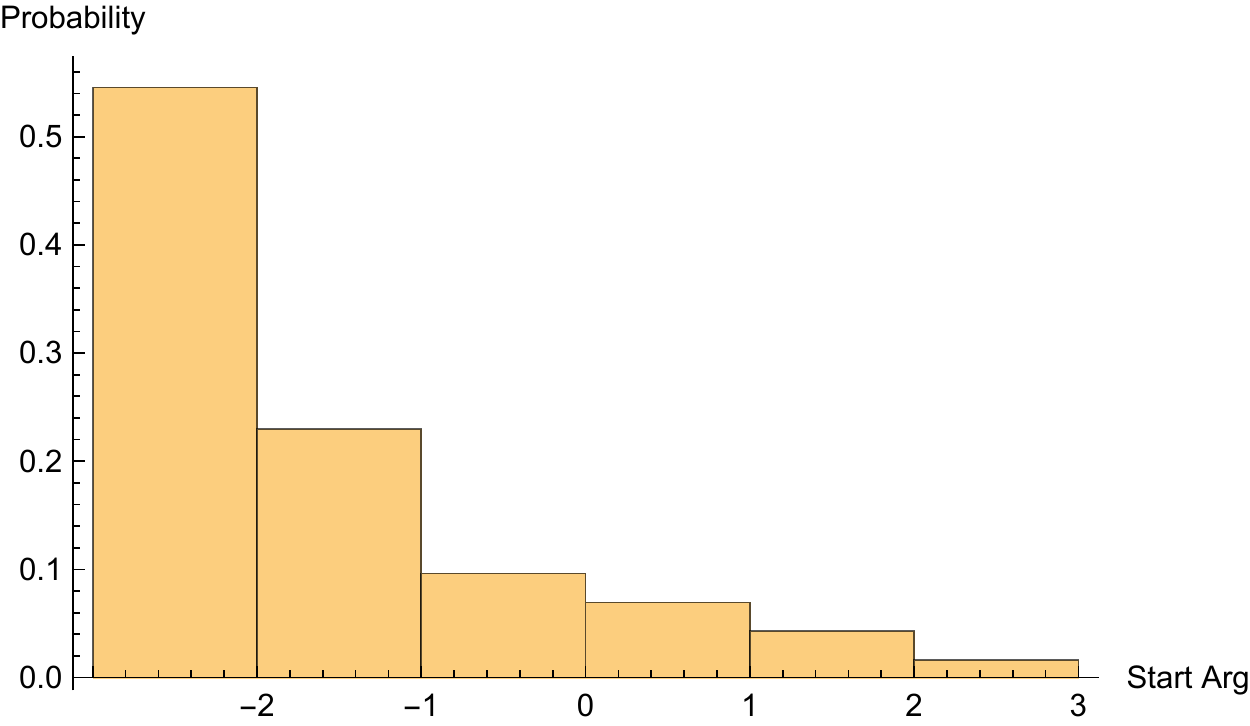}~~\includegraphics[width=3.5in]{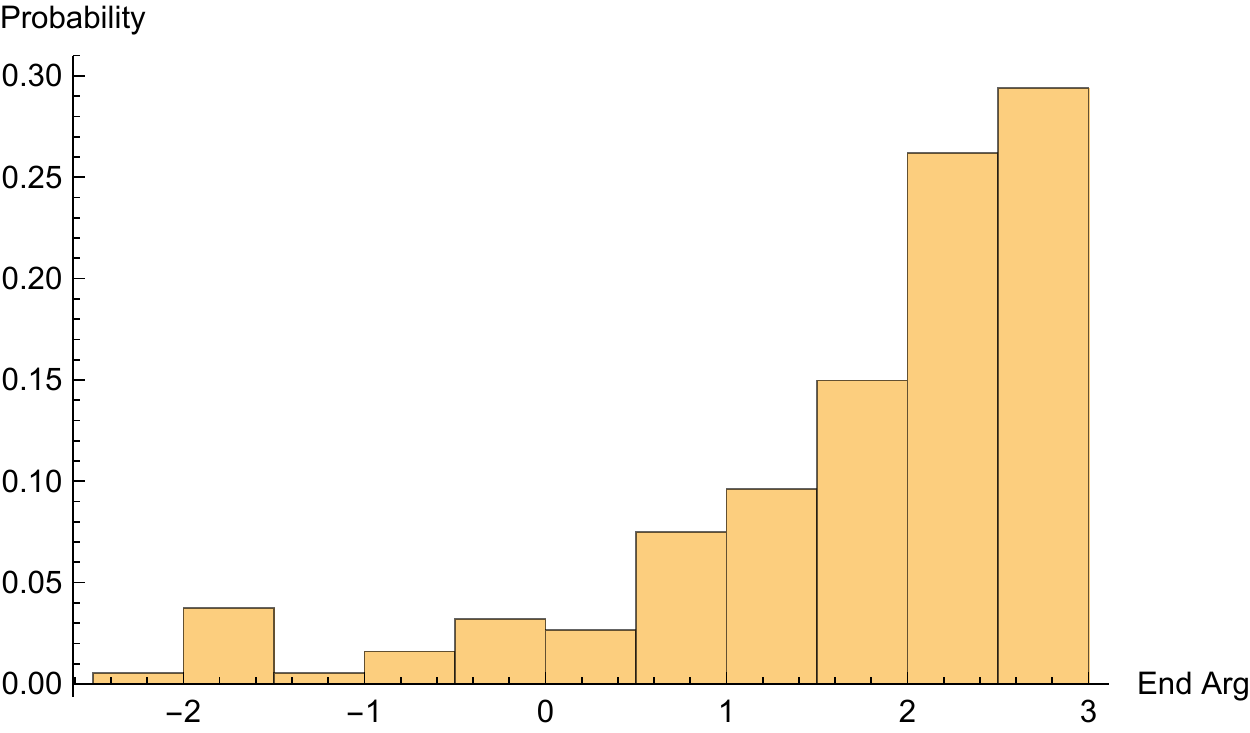}

\caption{(Above) Histogram of the length of island states for $t$ in the range 0-500.
(Below) Histograms of the quantity $\arg {\cal U}_{\cal K}$ at the beginning and end of island states  for $t$ in the range 0-500.}
\label{figisle3}
\end{figure}
The histograms of Fig. \ref{figisle3} illustrate the variations of the lengths of the island regions, and of the values of $\arg {\cal U}_{\cal K}(1,1;s)$ (in the range between
$-\pi$ and $\pi$) at the beginning and end of each island region for $t$ varying from 0-500. Around 60\% of islands have lengths below unity, over 50\% of starting arguments lie between -2 and -3, while over 50\% of end arguments lie between 2 and 3. The change  in argument values along the critical line and around the outer boundary of the island in $\sigma>1/2$ is $2\pi$ times the difference between the number of zeros and poles of ${\cal K}(1,1;s)$. In the case of Fig. \ref{figisland1} for example, the change of argument along the critical line is 4.4305, giving a change of argument along the outer boundary of 1.8401.

The filling of Table \ref{tabisle3} was somewhat labour intensive. A more automated procedure is possible if one wants to determine the fraction of zeros of a function in island regions; this can be done simply by applying equation (\ref{srep15}) to a table of zeros. Using this procedure, one can easily classify the number of zeros of $\zeta(s)$ lying on the critical line in island regions. For the first 10,000 zeros, this gives 7467 island zeros, with the fraction of zeros lying in islands varying little around 75\% for each set of 1000 zeros ranging from 1000 to 10000. For $L_{-4}(s)$, the first ten thousand zeros have 6925 lying within islands, and again the fraction lying within islands varies little from 70\% for sets of 1000 zeros ranging from 1000 to 10000.

\begin{figure}[h]
\includegraphics[width=3.5in]{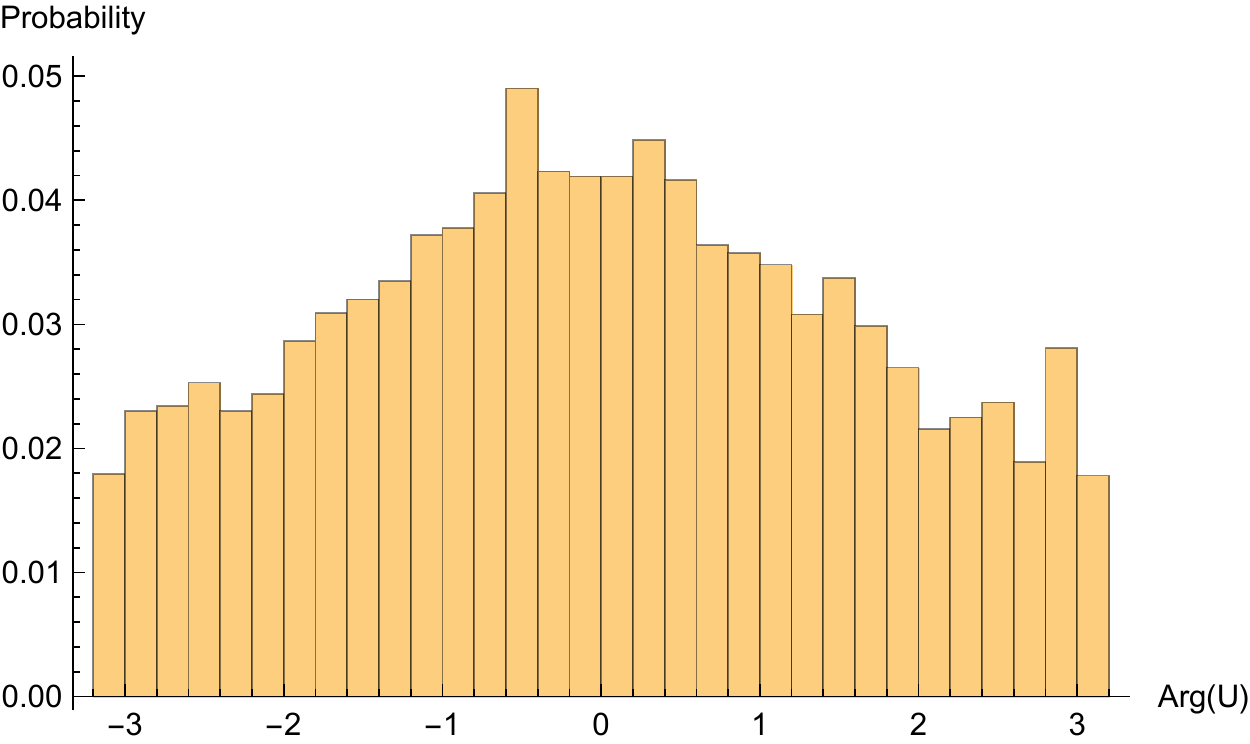}~~\includegraphics[width=3.5in]{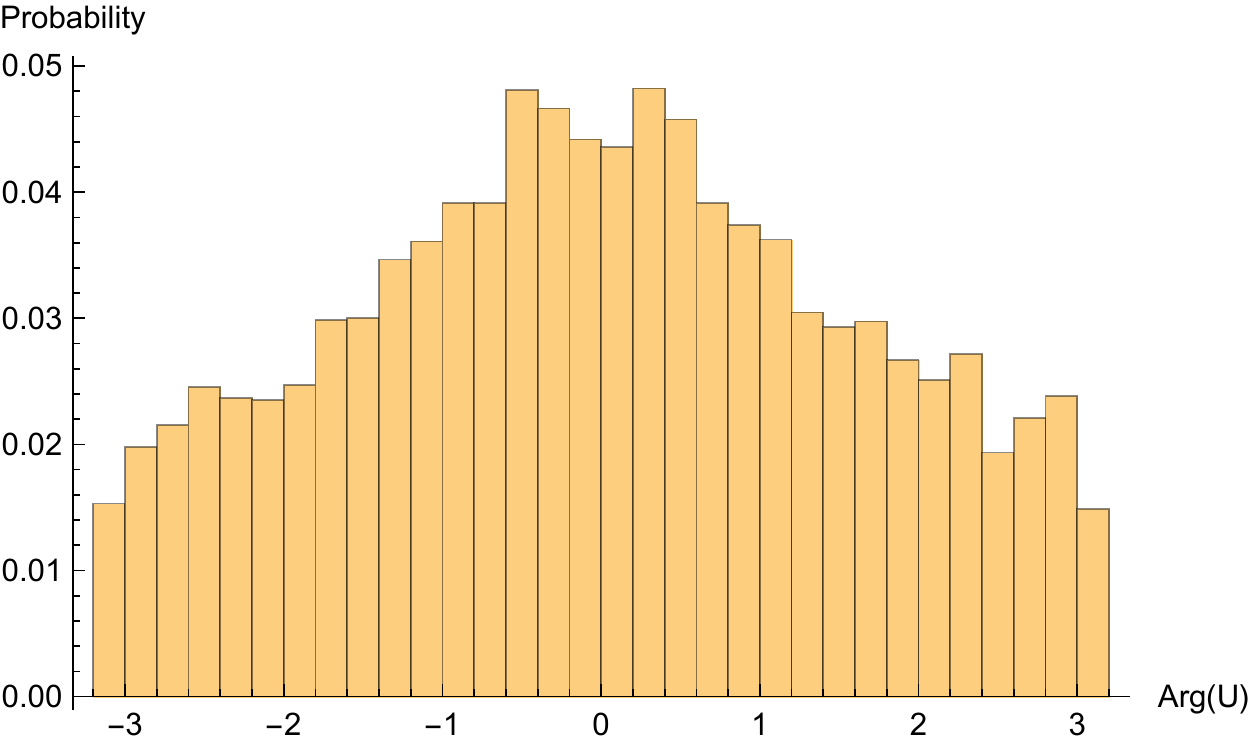}

\caption{ Histograms of the value of the argument of ${\cal U}(s)=\xi_1(2 s-1)/\xi_1(2 s)$ with $s$ chosen from the subset of the first ten thousand zeros of (left) $\zeta(s)$ or (right) $L_{-4}(s)$
which lie in island states.
}
\label{figisle4}
\end{figure}
The histograms of  Figure \ref{figisle4} show the distributions of the values of the argument of ${\cal U}(s)$ at the first ten thousand zeros of $\zeta (s)$ and $L_{-4}(s)$. These distributions are essentially flat over the range $-\pi$ to $\pi$. This is interesting, and perhaps surprising, in that the zeros correspond to values of $s$ for which 
${\cal U}_{\cal K}(1,1;s)=-{\cal U}(s)$; another set of zeros obeying the same equation is of course that of ${\cal L}(s)$, and these we have shown are concentrated in the neighbourhood of $\arg {\cal U}(s)=\pm \pi$.

We can analyse the zeros of $S_0(s)$ with regard to whether they lie in intervals of the critical line in which (\ref{srep15}) holds or does not hold. For each of these two sets, we can break them into clusters where the zeros of $\zeta(s)$ and $L_{-4}(s)$ both lie in unbroken consecutive sets. We have studied the first 10,000 zeros of $L_{-4}(s)$, which lie in an interval of $t$ where there are 8171 zeros of $\zeta(s)$. Of the 18171 zeros of $S_0(s)$,13014 or around 71.6\% lie in 2809 sequences where (\ref{srep15}) holds,
so that the average number of zeros per sequence is 4.633. Of these, the number of $\zeta$ zeros is on-average 2.167, and of $L_{-4}$ zeros  the average is 2.465-  these averages bearing a ratio of 0.8791:1. The contrary set of intervals have shorter sequences of zeros (average 1.8365), with the average numbers for $\zeta(s)$ being 0.7415 and for $L_{-4}(s)$
being 1.0947 (in the ratio 0.6773:1). Thus, the in-island intervals have a more balanced set of zeros of the two functions, with the intervening intervals having a greater fraction of zeros of $L_{-4}(s)$.

A useful result is obtained from equation (\ref{srep6}) if we solve for $S_0(1,s)$:
\begin{equation}
S_0(s)=\frac{\Gamma (s) S_0(s)}{8\pi^s}=\frac{{\cal V}_{\cal K}(1,1;s) {\cal T}_-(s)+{\cal T}_+(s)}{1-{\cal V}_{\cal K}(1,1;s) /(2 s-1)}.
\label{srep17}
\end{equation}
The representation (\ref{srep17}) may be forced into a form suitable for expanding about zeros of  $\tilde{S}_0(1,s)$:
\begin{equation}
S_0(s)=\frac{2 (2 s-1) {\cal T}_-(1,s)^2[ {\cal V}_{\cal K}(1,1;s;1)+{\cal V}(1,s)]} {{\cal L}(s)-2 {\cal T}_-(1,s)[ {\cal V}_{\cal K}(1,1;s)+{\cal V}(s)]}.
\label{srep18}
\end{equation}
When ${\cal L}(s)=0$, (\ref{srep18}) loses its dependence on $[ {\cal V}_{\cal K}(1,1;s;1)+{\cal V}(s)]$; otherwise, it is zero when this factor is zero.
When ${\cal T}_-(s)=0$, (\ref{srep17}) should be used;  the numerator then reduces to ${\cal T}_+(s)$, which is known to be non-zero if ${\cal T}_-$ is zero.
Equation (\ref{srep18}) can also be expressed in terms of $[ {\cal U}_{\cal K}(1,1;s)+{\cal U}(s)]$, since
\begin{equation}
[ {\cal V}_{\cal K}(1,1;s;1)+{\cal V}(1,s)]=\frac{[ {\cal U}_{\cal K}(1,1;s)+{\cal U}(1,s)]}{({\cal U}_{\cal K}(1,1;s)+1)(1-{\cal U}(1,s))}.
\label{srep19}
\end{equation}
Thus, $\tilde{S}_0(1,s)\rightarrow 0$ as $[ {\cal V}_{\cal K}(1,1;s)+{\cal V}(s)]\rightarrow 0$ or, equivalently, $[ {\cal U}_{\cal K}(1,1;s)+{\cal U}(s)]\rightarrow 0$.

In addition to (\ref{srep18}) and (\ref{srep18}), we have for the derivatives with respect to $s$:
\begin{equation}
[ {\cal V}'_{\cal K}(1,1;s)+{\cal V}'(s)]=\frac{2 {\cal U}'_{\cal K}(1,1;s)}{(1+{\cal U}_{\cal K}(1,1;s))^2}+\frac{2 {\cal U}'(s)}{(1-{\cal U}(s))^2}.
\label{srep20}
\end{equation}
Hence, at a zero $s_0$ of $\tilde{S}_0(1,s)$:
\begin{equation}
[ {\cal V}'_{\cal K}(1,1;s_0)+{\cal V}'(s_0)]=\frac{2 [{\cal U}'_{\cal K}(1,1;s_0)+{\cal U}'(s_0)]}{(1-{\cal U}'(s_0))^2}.
\label{srep21}
\end{equation}
For $s$ on the critical line,
\begin{equation}
\frac{d {\cal U}_{\cal K}(1,1;1/2+ i t)}{d s}={\cal U}_{\cal K}(1,1;1/2+ i t)\frac{d \arg {\cal U}_{\cal K}(1,1;1/2+ i t)}{d t}
\label{srep22}
\end{equation}
and
\begin{equation}
\frac{d {\cal U}(1/2+ i t)}{d s}={\cal U}(1/2+ i t)\frac{d \arg {\cal U}(1/2+ i t)}{d t}.
\label{srep23}
\end{equation}
Hence, for  $s_0$ on the critical line,
\begin{eqnarray}
\frac{d {\cal U}_{\cal K}(1,1;1/2+ i t_0)}{d s}+\frac{d {\cal U}(1/2+ i t_0)}{d s}={\cal U}(1/2+ i t_0)  && \nonumber \\
\left[\frac{d \arg {\cal U}(1,1/2+ i t_0)}{d t}-\frac{d \arg {\cal U}_{\cal K}(1,1;1/2+ i t_0)}{d t}\right]  & &.
\label{srep24}
\end{eqnarray}

{\bf Remark:} We know that all zeros of $\tilde{S}_0(s)$ in extended regions lie on the critical line. From equation (\ref{srep24}), we then also know that all these zeros are simple, since  the argument derivatives of ${\cal U}_{\cal K}$ and ${\cal U}$ there have opposite signs.

\begin{figure}[h]
\includegraphics[width=3.5in]{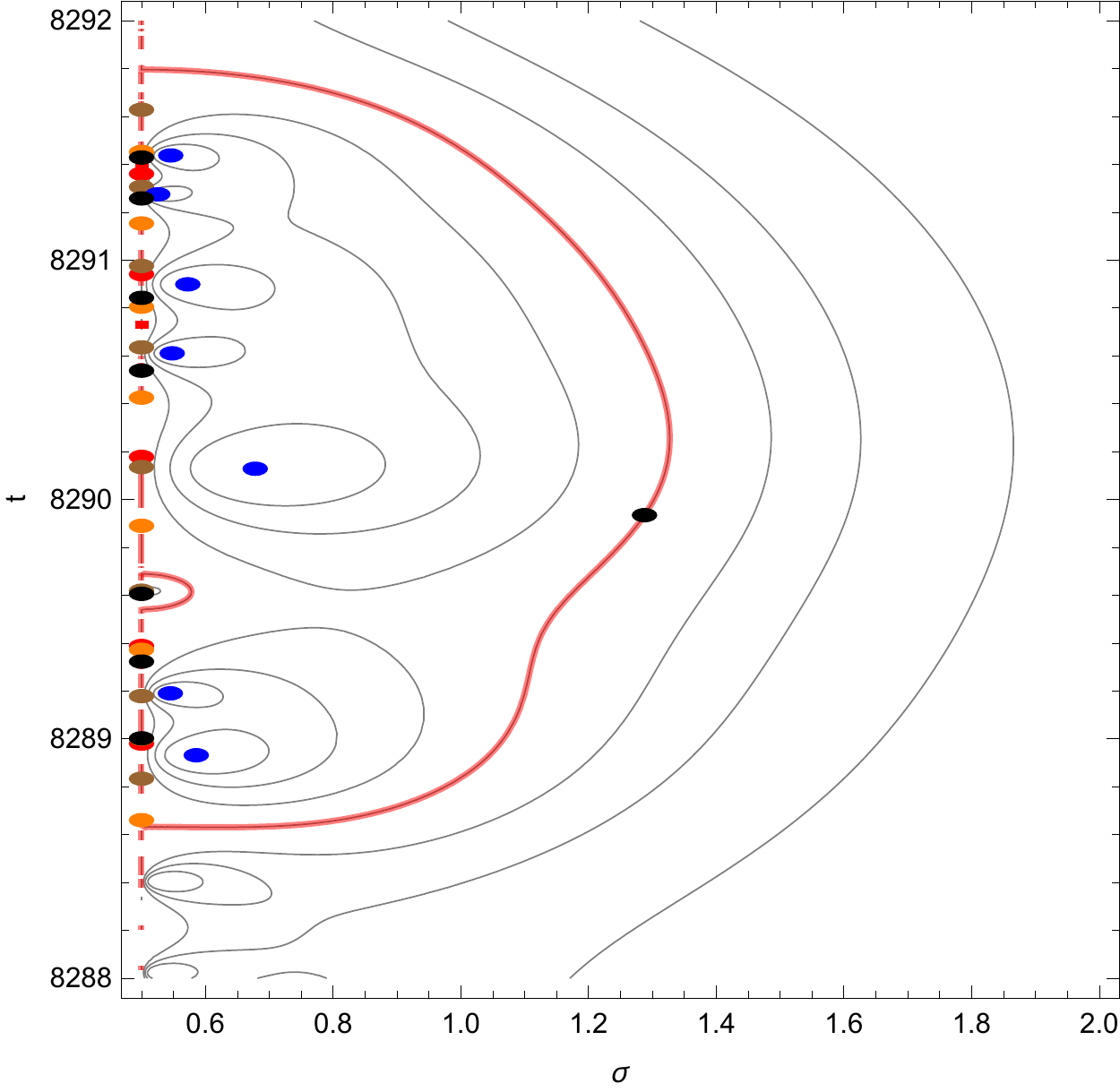}\\
\includegraphics[width=3.5in]{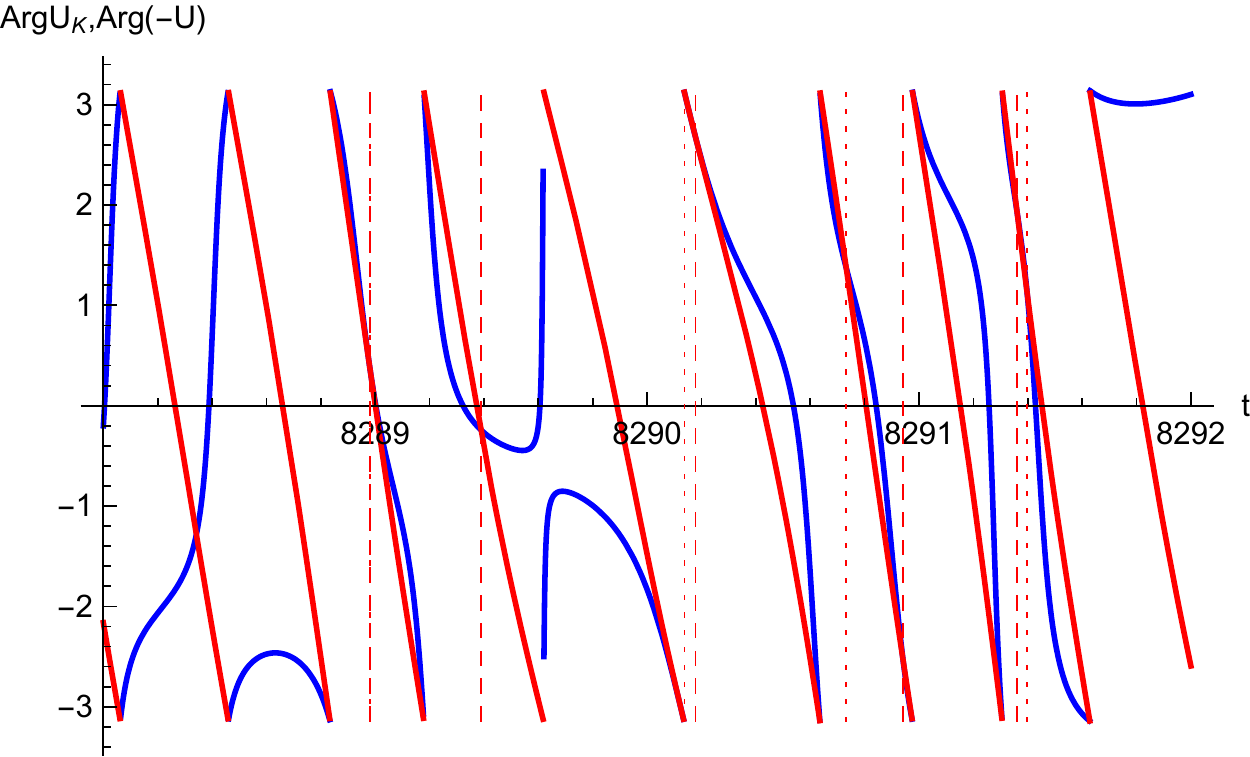}
\caption{(Above) Contours of $|{\cal  U}_{\cal K}(1,1; \sigma+i t)|$ in the plane $(\sigma, t)$, with the red contour corresponding to modulus unity.  The coloured dots correspond to
zeros of $\zeta (s)$ (red), ${\cal  U}_{\cal K}(0,0; \sigma+i t)$ (black),  ${\cal  U}_{\cal K}(1,1; \sigma+i t)$ (blue), $ {\cal K}_\lambda(0,0; \sigma+i t)$ (green), ${\cal T}_+(\sigma+i t)$ (orange) and ${\cal T}_-(\sigma+i t)$ (brown); red rectangles denote zeros of $L_{-4}(s)$. 
(Below)  $\arg{\cal  U}_{\cal K}(1,1; 1/2+i t)$ (blue curve) and $\arg (-{\cal U}(1/2+i t))$ (red curve) as a function of $t$. The island region starts  at $t=8288.63233$ and ends at $t=8291.79597$.  Red dashed lines for chosen values of $t$ indicate zeros of $\zeta(s)$, while red dotted lines indicate zeros of 
$L_{-4}(s)$.}
\label{figisland5}
\end{figure}

Returning to the equations (\ref{mthm2}) and (\ref{srep10}), we have that zeros of $S_0(s)$ and ${\cal L}(s)$ lie on lines where $|{\cal U}_{\cal K}(1,1;s)|=|{\cal U}(s)|$.
For island regions, the contours $|{\cal U}_{\cal K}(1,1;s)/{\cal U}(s)|=1$ lie  on separate lines surrounding each of the zeros of ${\cal U}_{\cal K}(1,1;s)$ in $\sigma>1/2$ and poles in $\sigma<1/2$. We will call the regions within islands where  $|{\cal U}_{\cal K}(1,1;s)/{\cal U}(s)|>1$ the {\em outer islands}, and the regions where $|{\cal U}_{\cal K}(1,1;s)/{\cal U}(s)|=1$ the {\em inner islands}.

{\bf Remark:} All zeros of $\tilde{S}_0(s)$ not on the critical line must lie on the boundaries between outer and inner islands. 

We define 
\begin{equation}
{\cal F}(s)=\frac{{\cal U}_{\cal K}(1,1;1;s)}{{\cal U}(s)}=\frac{1+{\cal G}(s)}{1-{\cal G}(s)},~~{\cal G}(s)=\frac{{\cal F}(s)-1}{{\cal F}(s)+1}.
\label{srep25}
\end{equation}
Then the boundaries between inner and outer islands are given by $|{\cal F}(s)|=1$.  In terms of argument derivatives on the critical line, then
the respective conditions for the extended, outer island and inner island regions are
\begin{equation}
\frac{\partial \arg {\cal U}_{\cal K}(1,1;1;1/2+ it)}{\partial t}>0,~~\frac{\partial \arg {\cal F}(1/2+ it)}{\partial t}>0,~~\frac{\partial \arg {\cal F}(1;1/2+ it)}{\partial t}<0.
\label{srep26}
\end{equation}

In Tables \ref{tabisle4} and \ref{tabisle5} we show the results of classifying  the first ten thousand zeros of $\zeta (s)$ and $L_{-4}(s)$  using  the conditions (\ref{srep26}).

\begin{table}
\begin{tabular}{|c|c|c|c|c|}\hline
 $t$ Range&Inner island zeros& All zeros &1-$\zeta(ii)/\zeta(all)$\\ \hline
0-1000& 173& 649 & 0.7344\\
1000-2000& 224&868&0.7419\\
2000-3000& 245& 952 & 0.7426\\
3000-4000& 284& 1005 & 0.7174\\
4000-5000 & 290 &1046 &0.7228\\
5000-6000&301 &1078 &0.7208\\ 
6000-7000 & 298 & 1105 &  0.7303 \\  
7000-8000 & 314 &1127 &0.7214\\
8000-9000&334 &1148 &0.7091\\ \hline
9000-9877.78 & 286 & 1021 &  0.7199 \\ \hline 
\end{tabular}
\caption{Statistics of the first 10,000 zeros of $\zeta (s)$.}
\label{tabisle4}
\end{table}

For $\zeta(s)$, the fourth column of Table \ref{tabisle4} shows the fraction of zeros which do not lie within the inner islands. We can take this as a proxy for the fraction of zeros
which we know to lie on the critical line.  The mean value of the fraction for the zeros lying up to $t=8000$ is 0.7266, with the standard deviation being 0.0113.This fraction of course is numerically rather than analytically determined, but it is of interest to compare it with the results established by analytic means for the fraction of zeros proven to lie on the critical line. These have progressed from 1/3 (Levinson, 1974)  to 2/5 (Conrey, 1989) and 41\% (Bui, Conrey and Young, 2011).  This comparison shows the interest in analytic proof of
the numerical estimation presented here.

\begin{table}
\begin{tabular}{|c|c|c|c|c|}\hline
 $t$ Range&Inner island zeros& All zeros &1-$L_{-4}(ii)/L_{-4}(all)$\\ \hline
0-1000& 164& 868 & 0.8111\\
1000-2000& 260&1090&0.7615\\
2000-3000& 318& 1172 & 0.7287\\
3000-4000& 305& 1226 & 0.7512\\
4000-5000 & 335 &1267 &0.7356\\
5000-6000&374 &1298 &0.7119\\ 
6000-7000 & 385 & 1326 &  0.7097 \\  
7000-8000 & 357 &1347 &0.7350\\ \hline
8000-8297.64&103 &406&0.7463\\ \hline
\end{tabular}
\caption{Statistics of the first 10,000 zeros of $L_{-4} (s)$.}
\label{tabisle5}
\end{table}

The results for $L_{-4}(s)$ in Table \ref{tabisle5} are quite similar to those for $\zeta (s)$, with the mean fraction of zeros up to $t=8000$ being 0.7431, with a standard deviation of 0.0326. The mean fraction for$\zeta(s)$ lies within one standard deviation of that for  $L_{-4}(s)$.

\begin{figure}[h]
\includegraphics[width=3.0in]{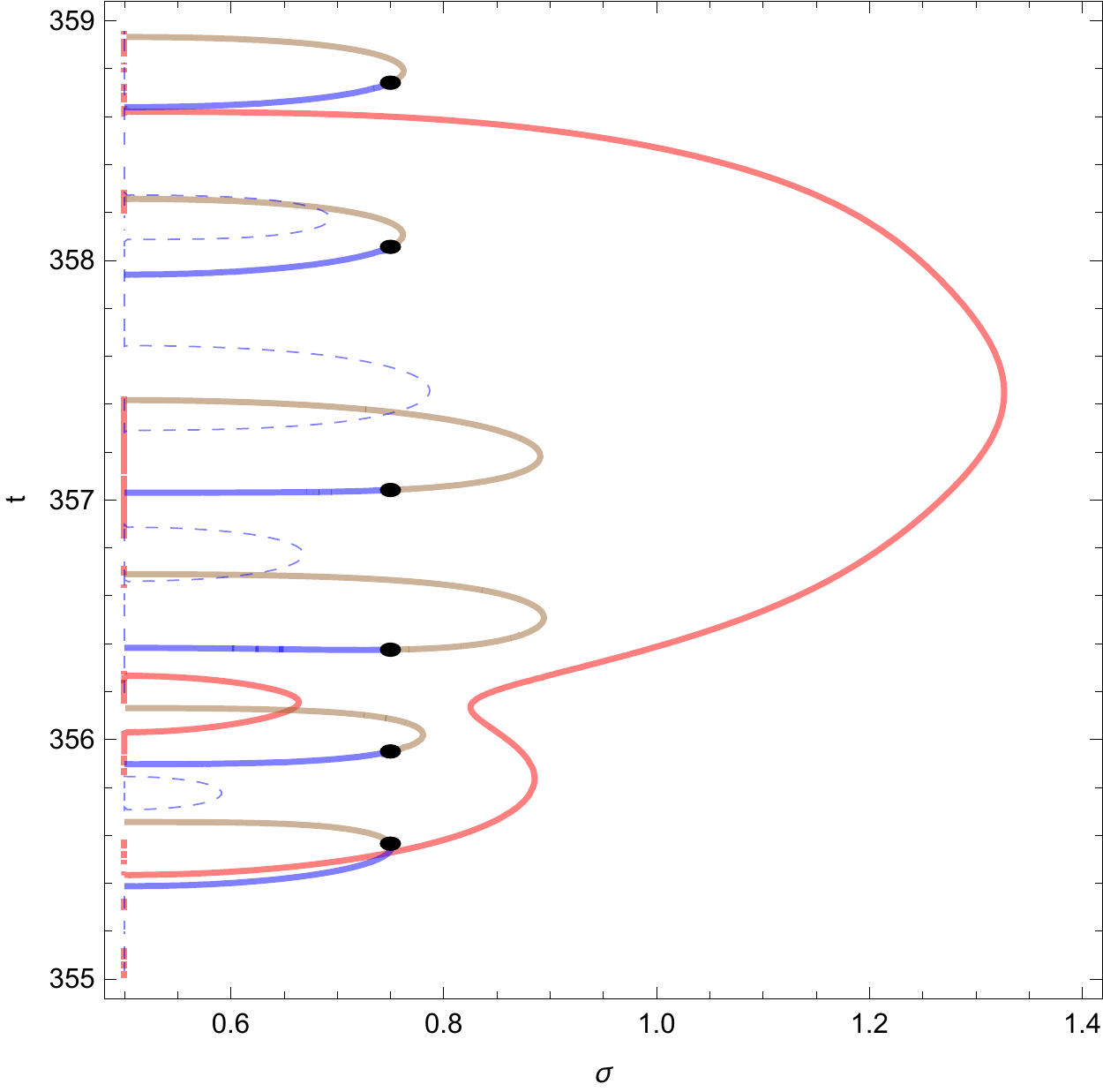} ~~\includegraphics[width=3.0in]{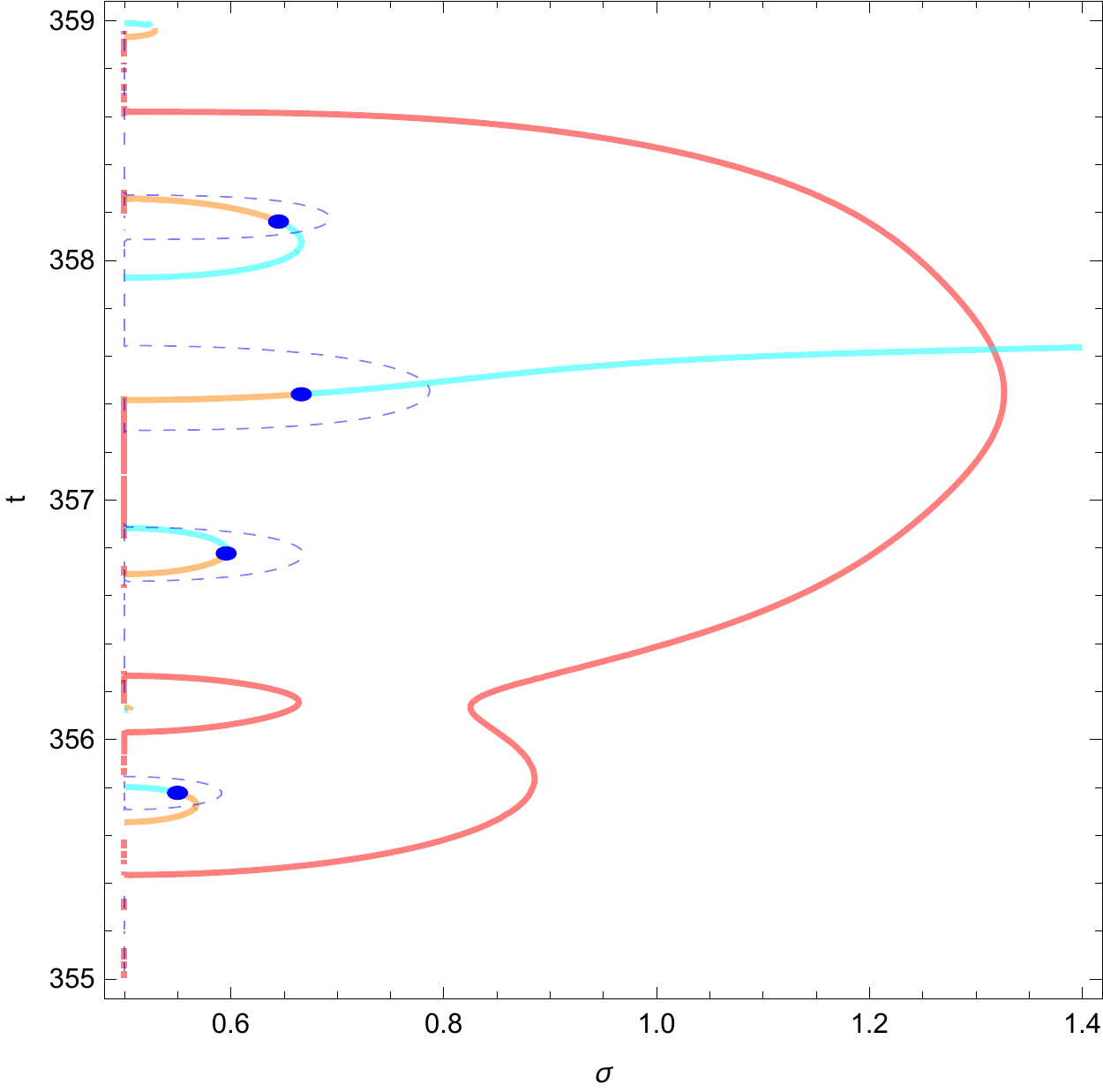} 
\includegraphics[width=3.0in]{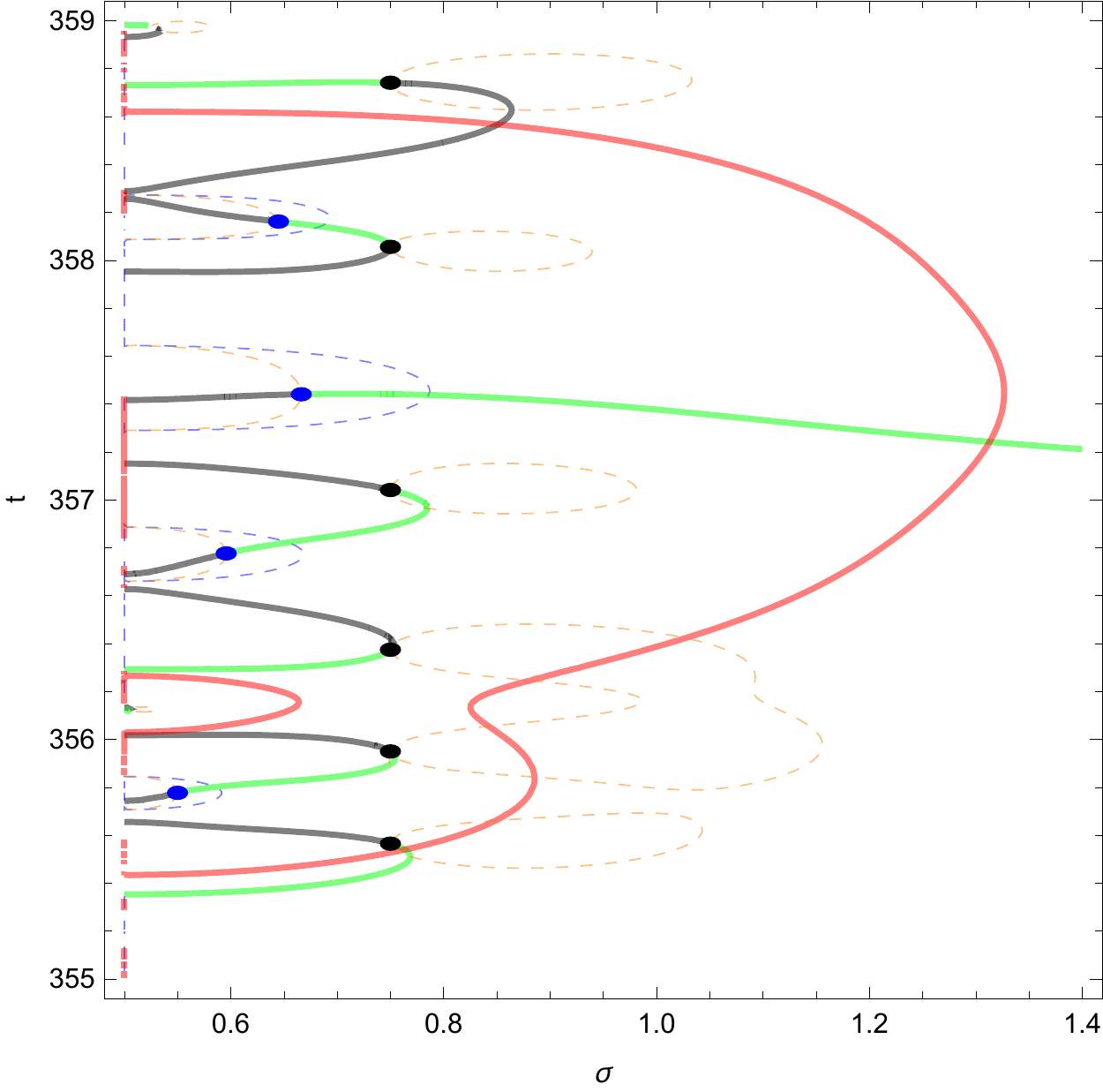}~~\includegraphics[width=3.0in]{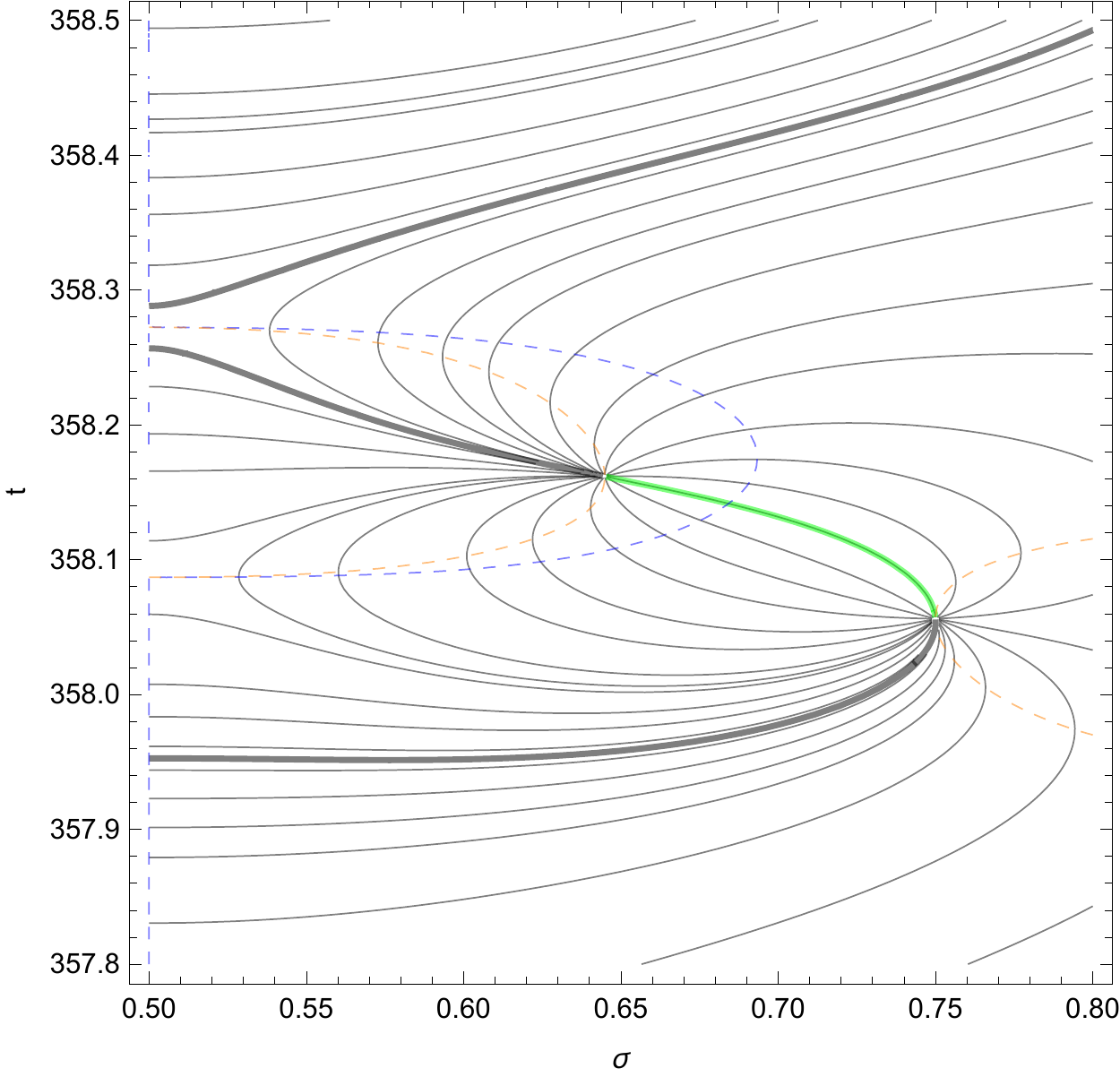}
\caption{(Above,left and right ) Contours of  $\arg{\cal  U}(\sigma+i t)$  and $\arg{\cal  U}_{\cal K}(1,1; \sigma+i t,1)$ in the plane $(\sigma, t)$, with the red contour corresponding to modulus unity for $|{\cal  U}_{\cal K}(1,1; \sigma+i t,1)|$, and the dashed blue contour corresponding to modulus unity for
$|{\cal  U}_{\cal K}(1,1; \sigma+i t,1)/{\cal  U}_{\cal K}(\sigma+i t)|$. At left, brown and blue contours correspond to arguments $0$ and $\pi$ respectively, and at right
aquamarine to 0 and orange to $\pi$.
 The coloured dots correspond to
zeros  of  $\zeta(2 s-1)$ (black) and  ${\cal  U}_{\cal K}(0,0; \sigma+i t,1)$ (blue).
(Below, left and right)  Contours of $\arg{\cal  U}_{\cal K}(1,1; \sigma +i t,1)/{\cal  U}(\sigma+i t)$ (0, green, $\pi$, black curve) at left, with the dashed orange curves corresponding
to the $t$ derivative of the argument being zero. At right, the contours show extra detail of the variation of argument for a region near the top of the island.}
\label{figisland6}
\end{figure}

We now consider the properties of lines of constant argument zero and $\pi$ for the functions ${\cal U}(s)$, ${\cal U}_{\cal K}(1,1;s;1)$ and their ratio ${\cal F}(s)$. These are shown in Fig. \ref{figisland6}, for the same range of $t$ as Fig.\ref{figisland5}. 

For ${\cal U}(s)$, the behaviour of the lines is simple and does not depend on whether $s$ lies inside or outside the island region. They leave the critical line, pass through a zero of $\zeta (2 s-1)$ where they undergo an argument change of $\pi$, and rejoin the critical line. There is a one-to-one correspondence between zeros of $\zeta (2 s-1)$ and lines of argument $\pi$ reaching the critical line. The points where the trajectories reach the critical line are zeros of ${\cal T}_-(s)$ (argument $0$), or ${\cal T}_+(s)$ (argument $\pi$).

For ${\cal U}_{\cal K}(1,1;s)$, the behaviour of the lines is slightly more complicated. Lines of constant argument now pass through zeros of ${\cal K}(1,1;s)$, where they undergo an argument change of $\pi$ before returning to the critical line. The points where the trajectories reach the critical line are zeros of ${\cal K}(0,0;s;1)(s)$ (argument $0$), or ${\cal K}_\lambda(0,0;s)(s)$  (argument $\pi$). One exceptional trajectory runs from the critical line with an argument of $\pi$, passes through a zero
of ${\cal K}(1,1;s)$, emerges with an argument of $0$ and cuts the island boundary at a zero of ${\cal K}(0,0;s)(s)$. Taking into account the four zeros and one pole of ${\cal K}(1,1;s)$ in the island region for $\sigma>1/2$, the change of argument of ${\cal U}_{\cal K}(1,1;s)$ round the closed contour bounding the island in $\sigma \ge 0$ is $6\pi$.

The behaviour of the argument of ${\cal F}(s)$ is more complicated than that of either of its constituents. From Fig. \ref{figisland6}, there are ten lines with $\arg{\cal F}(s)=\pi$ cutting the critical line within the island region. Two are associated with a pole of ${\cal F}(s)$, one at the upper end of the island and the second at the lower end. A third is associated with
a pole of ${\cal F}(s)$ and a line with $\arg{\cal F}(s)=0$ cutting the boundary of the island region. In order of increasing $t$, the ten lines of argument $\pi$ are associated with zeros of the following functions: ${\cal L}$, $L_{-4}$,  $\zeta$, $L_{-4}$, ${\cal L}$, $\zeta$, ${\cal L}$, $\zeta$, ${\cal L}$, $L_{-4}$ (see Table \ref{tabisle2}).

Outside inner island regions, ${\cal U}_{\cal K}(1,1;s)$ is alike   ${\cal U}(s)$, in that both functions have zeros to the right of the critical line, and poles to its left. In consequence,
${\cal F}(s)$ has contours of piecewise-constant argument which proceed from the critical line in $\sigma>1/2$ to either a zero of ${\cal U}_{\cal K}(1,1;s)$ or   ${\cal U}(s)$, and then return to the critical one with an argument changed by $\pi$. The consequence of this is that lines of constant argument $\pi$ reaching the critical line may be placed in one-to-one correspondence with the zeros of either ${\cal K}(1,1;s)$ or $\zeta (2 s-1)$ outside inner island regions. The former correspond to zeros of either ${\cal L}(s)$ or $S_0(s)$, which we know must lie on the critical line outside inner island regions. We also know that the zeros of ${\cal L}(s)$ may be placed in one-to-one correspondence with the zeros of ${\cal T}_-(s)$ and
$\zeta(2 s-1)$ (Ki,2006, Lagarias and Suzuki, 2006, McPhedran and Poulton, 2013), for all values of $t$. However, lines of constant argument may link a point on the critical line outside the island region with a zero off the critical line inside the island region, or a point inside an inner island region to a zero in an outer island region; examples of both are shown in Fig. \ref{figisland6}.

Inside inner island regions, lines of constant argument are influenced by  zeros of ${\cal K}(1,1;s;1)$, and zeros  of $\zeta (2 s-1)$, which must lie outside the inner islands (being poles of ${\cal F}(s)$). Lines of constant argument thus run from the former to the latter, in some cases directly, and in others passing through a point where $\partial {\cal F}(s)/\partial t=0$. The trajectories of such turning points shown in Fig. \ref{figisland6} run from the zero of ${\cal K}(1,1;s;1)$ towards the critical line in one case, and from the zero of $\zeta (2 s-1)$ away from it in the other case. The trajectories $\partial {\cal F}(s)/\partial t=0$ are tangent to the trajectories $|{\cal F}(s)|=1$ at the critical line.

We can strengthen a previous remark, as follows.

{\bf Remark:} $S_0(s)$ has all its zeros on the critical line if and only if lines of argument $\pi$ run from zeros of ${\cal U}_{\cal K}(1,1;s;1)$ inside inner islands to the critical line without intersecting the boundaries of the inner island (the lines where $|{\cal F}(s)|=1$).

The argument of ${\cal F}(s)$ increases monotonically round the boundary of each inner island, so the condition just enunciated is also equivalent to the requirement that
$|{\cal F}(s)|=1$ and $\arg{\cal F}(s)=\pi$ only hold simultaneously on the critical line.

\begin{figure}[h]
\includegraphics[width=3.0in]{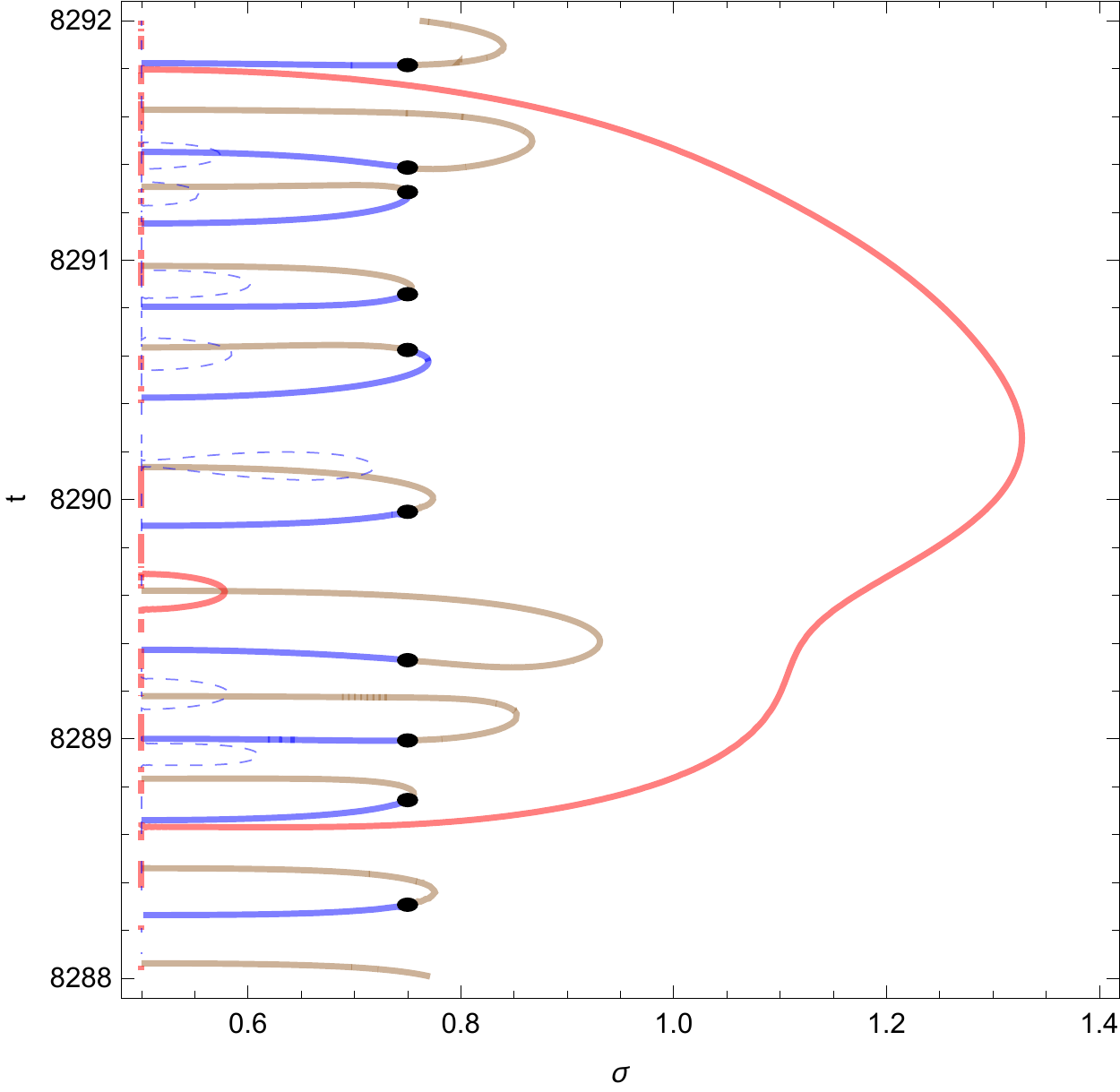} ~~\includegraphics[width=3.0in]{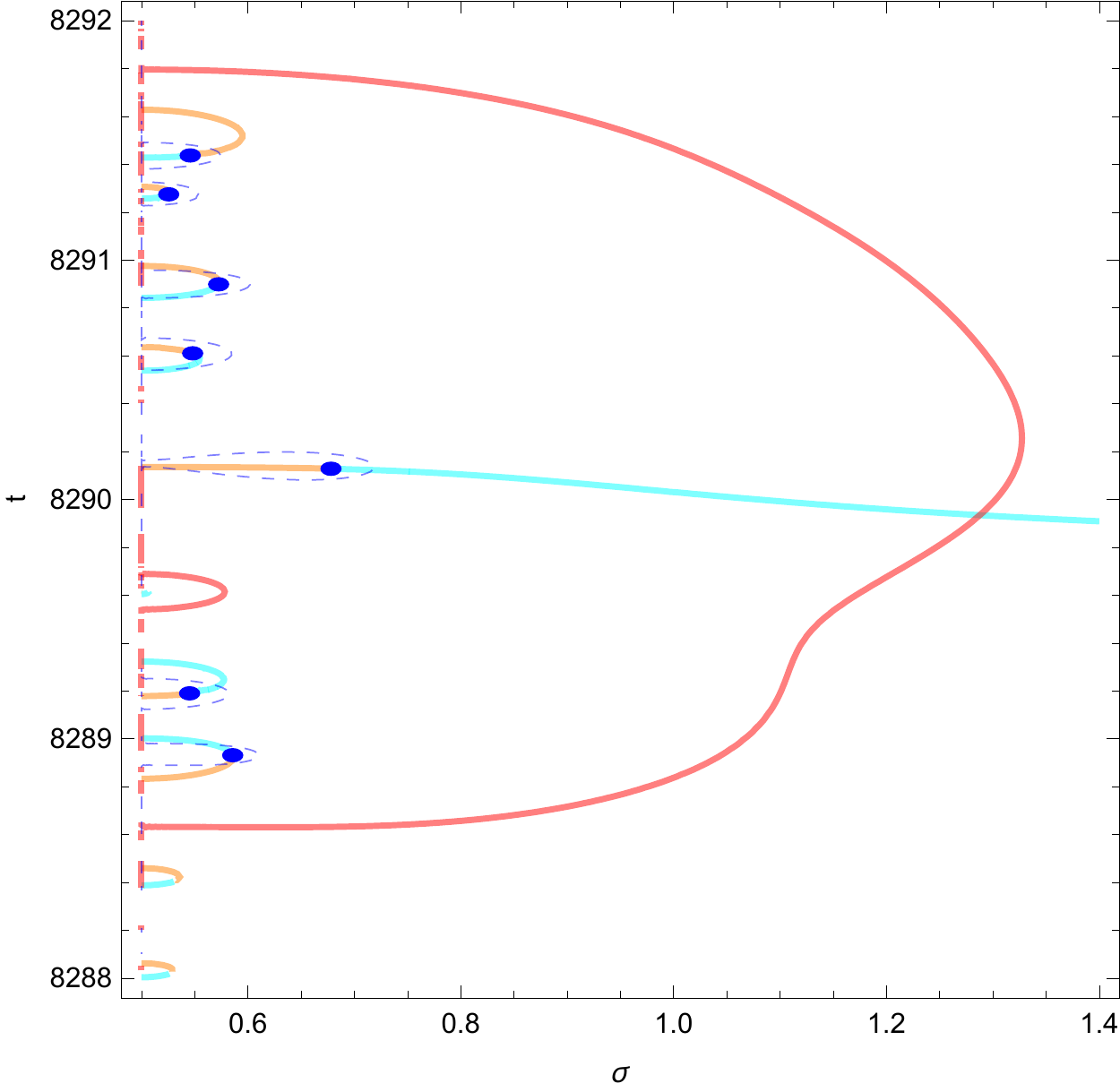} 
\includegraphics[width=3.0in]{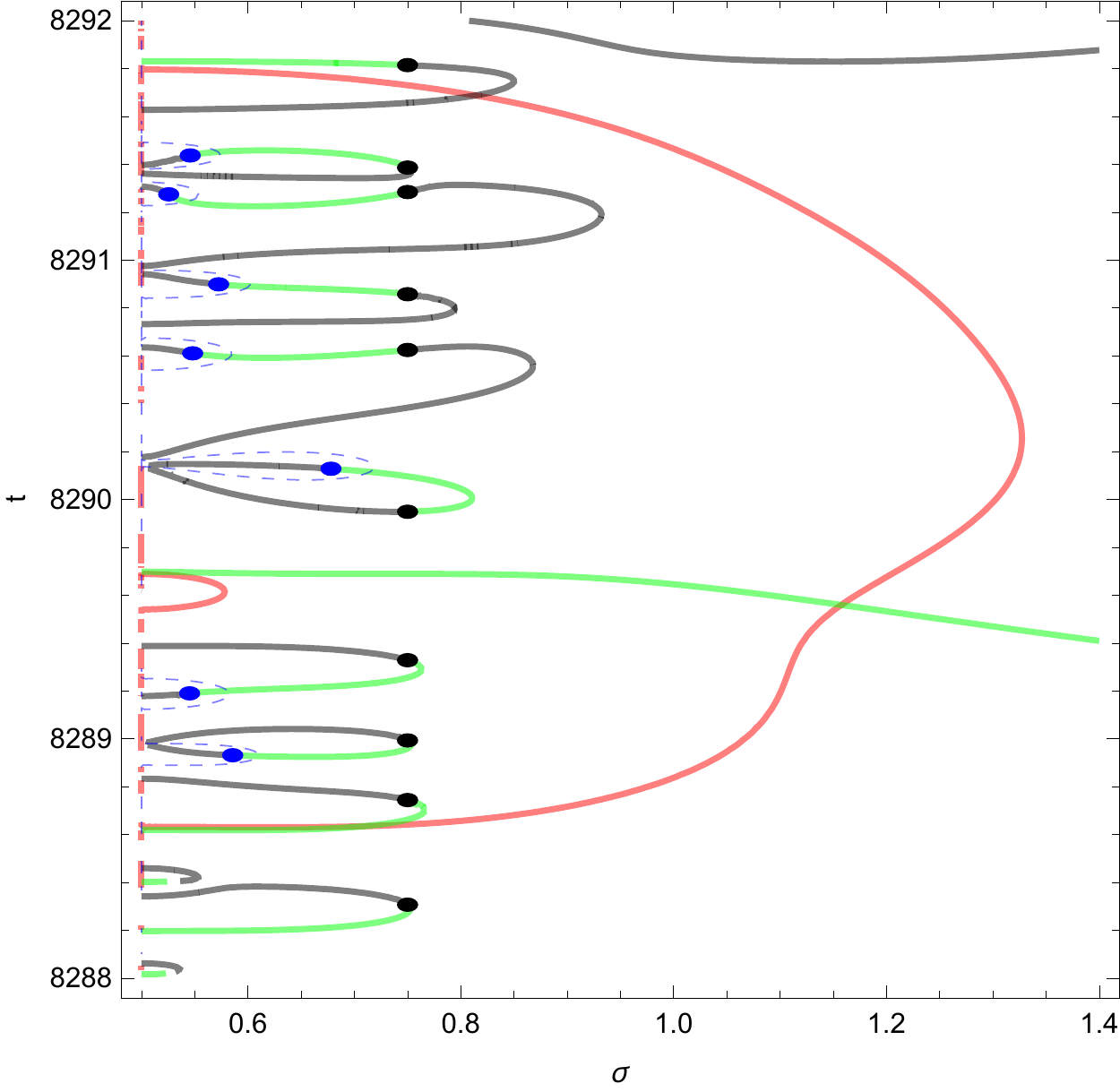}~~\includegraphics[width=3.0in]{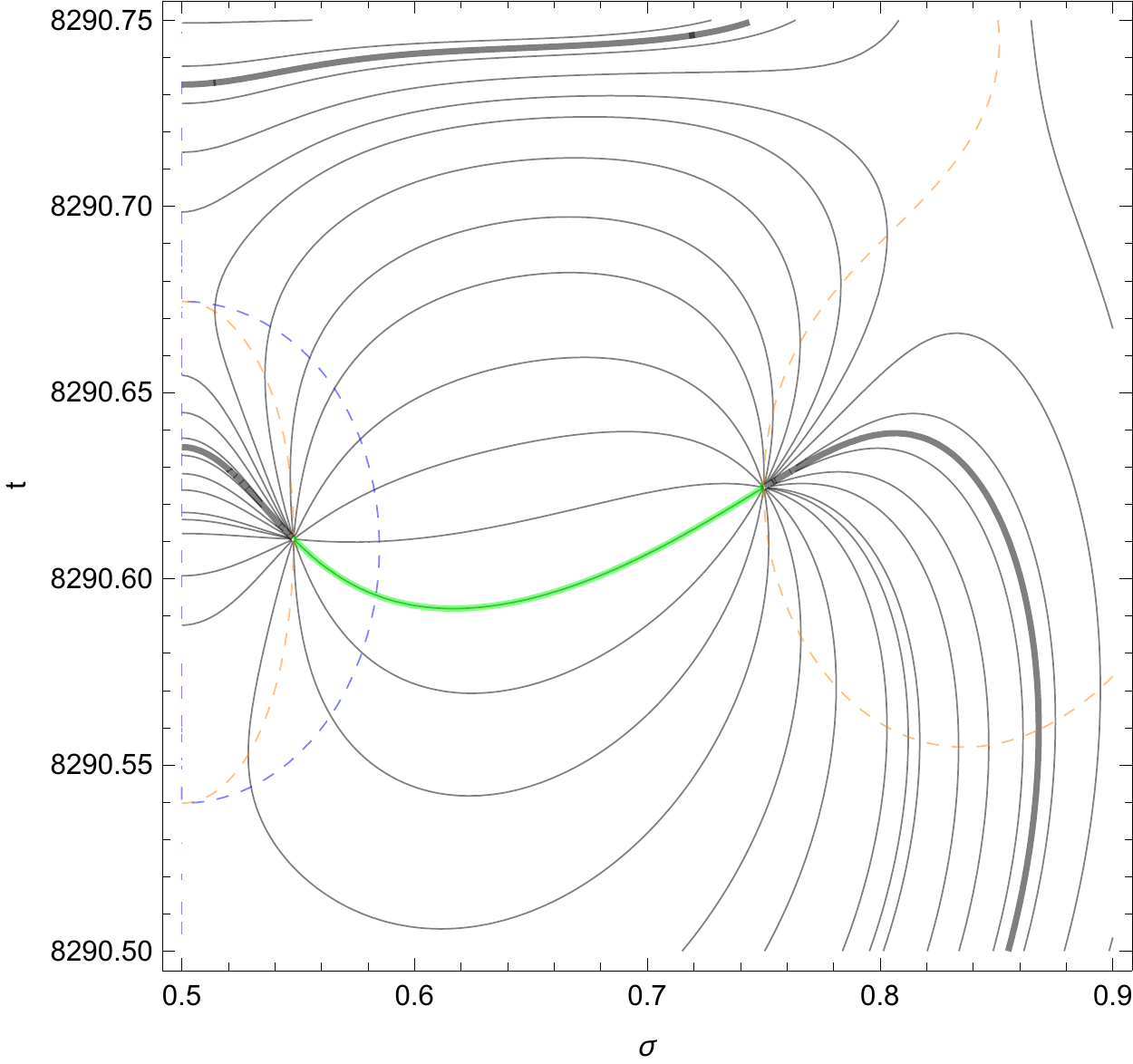}
\caption{(Above,left and right ) Contours of  $\arg{\cal  U}(\sigma+i t)$  and $\arg{\cal  U}_{\cal K}(1,1; \sigma+i t)$ in the plane $(\sigma, t)$, with the red contour corresponding to modulus unity for $|{\cal  U}_{\cal K}(1,1; \sigma+i t)|$, and the dashed blue contour corresponding to modulus unity for
$|{\cal  U}_{\cal K}(1,1; \sigma+i t)/{\cal  U}_{\cal K}(\sigma+i t)|$. At left, brown and blue contours correspond to arguments $0$ and $\pi$ respectively, and at right
aquamarine to 0 and orange to $\pi$.
 The coloured dots correspond to
zeros  of  $\zeta(2 s-1)$ (black) and  ${\cal  U}_{\cal K}(0,0; \sigma+i t)$ (blue).
(Below)  Contours of $\arg{\cal  U}_{\cal K}(1,1; \sigma +i t)/{\cal  U}(\sigma+i t)$ (0, green, $\pi$, black curve), with the dashed orange curves corresponding
to the $t$ derivative of the argument being zero. At right, the contours show extra detail of the variation of argument for a region above the middle of the island.}
\label{figisland7}
\end{figure}
In Fig. \ref{figisland7} we show curves for an island region at a far larger value of $t$, for comparison with Fig. \ref{figisland6}. The general forms of the two figures are similar: some details are changed less than others. For example, the $\sigma$ and $t$ ranges of the two islands are similar, but the numbers of zeros in each are quite different: seven zeros of ${\cal K}(1,1;s)$ plus one enclave zero for the former, by comparison with four and one for the latter. The total number of zeros of $S_0(s)$ and ${\cal L}(s)$ for the former is 17, and 11 for the latter. The $\sigma$ range  of the inner islands tends to be smaller in Fig. \ref{figisland7}  than in Fig. \ref{figisland6}, with the consequence that the distances between adjacent  zero and poles of ${\cal U}_{\cal K}(1,1;s)$ diminishes in comparison with the distance between zeros of  ${\cal U}_{\cal K}(1,1;s)$ and the closest zero of
${\cal U}(s)$.

Before giving the most important analytic arguments of this paper, some definitions are necessary. Number the inner islands with an integer $m$, and let the $m$th inner island segment on the critical line run from $t_l^{(m)}$ up to $t_u^{(m)}$. Let
\begin{equation}
\mu_l^{(m)}=\arg [-{\cal F}(1/2+i t_l^{(m)})],~\mu_u^{(m)}=\arg [-{\cal F}(1/2+i t_u^{(m)})].
\label{fthm1}
\end{equation}
\begin{theorem} The Riemann Hypothesis for $S_0(s)$ holds if and only if $\mu_u^{(m)}<0$, $\mu_l^{(m)}>0$ for all $m$.
\label{lthm}
\end{theorem}
\begin{proof}
If $\mu_u^{(m)}<0$, $\mu_l^{(m)}>0$ for a given $m$, then, as $\arg [-{\cal F}(1/2+i t)]$ increases monotonically as $t$ decreases inside the inner island, there is a point on the critical line between  $t_u^{(m)}$ and  $t_l^{(m)}$ where  $\arg [-{\cal F}(1/2+i t)]=0$. This is then the single point on the boundary $\Gamma_+^{(m)}$ of that part of the inner island in $\sigma\ge 1/2$ where either $S_0(s)=0$ or ${\cal L}(s)=0$. Given this holds for all $m$, all zeros of $S_0(s)$ and ${\cal L}(s)$ in inner islands lie on the critical line. We also know that
all zeros of $S_0(s)$ in enclaves or outside inner islands (i.e. away from the boundaries of inner islands) lie on the critical line, completing this part of the proof.
If the Riemann Hypothesis holds for $S_0(s)$, we know it also holds for ${\cal L}(s)$. Every inner island $m$ has one part on its boundary $\Gamma_+^{(m)}$ where $\arg [-{\cal F}(s)=0$: this must lie between  $t_u^{(m)}$ and  $t_l^{(m)}$. Given $\arg [-{\cal F}(s)]$ increases as one goes from the former to the latter, then $\mu_u^{(m)}<0$ and  $\mu_l^{(m)}>0$.
\end{proof}
\begin{corollary}
If between every two inner islands there exists at least one point on the critical line where $\arg[-{\cal F}(s)]=0$, i.e. one point on the critical line where either $S_0(s)=0$ or
${\cal L}(s)=0$, then the Riemann Hypothesis holds for $S_0(s)$.
\label{c1lthm}
\end{corollary}
\begin{proof}
Given an inner island $m$, it has at least one point with $t>t_u^{(m)}$ where $\arg [-{\cal F}(s)]=0$, and one with $t<t_l^{(m)}$. Going from the nearest such point above down to
 $t_u^{(m)}$, $\arg [-{\cal F}(s)]$ decreases, and so $\arg[-{\cal F}(1/2+i t_u^{(m)})]<0$. Going from the nearest such point below up to
 $t_l^{(m)}$, $\arg [-{\cal F}(s)]$ increases, and so $\arg[-{\cal F}(1/2+i t_l^{(m)})]>0$. Thus, for all $m$, $\mu_u^{(m)}<0$ and  $\mu_l^{(m)}>0$, so the Riemann Hypothesis holds
 for $S_0(s)$.
\end{proof}
\begin{corollary}
If the Riemann Hypothesis holds for $S_0(s)$, then between every two inner islands there exists at least one point on the critical line where $\arg[-{\cal F}(s)]=0$, i.e. one point on the critical line where either $S_0(s)=0$ or
${\cal L}(s)=0$.
\label{c2lthm}
\end{corollary}
\begin{proof}
If the Riemann Hypothesis holds, then $\mu_u^{(m)}<0$, $\mu_l^{(m)}>0$ for all $m$. Thus, for every inner island $m$, $\mu_u^{(m)}<0$ and $\mu_l^{(m+1)}>0$. Thus, there exists
at least one point on the critical line between $t_u^{(m)}$ and $t_l^{(m+1)}$ where $\arg[-{\cal F}(s)]=0$.
\end{proof}

{\bf Acknowledgement:} The author acknowledges helpful comments from Dr. Masatoshi Suzuki on early versions of this work.

\end{document}